\documentclass[11pt,reqno]{amsart}

\usepackage{color}
\usepackage[colorlinks=true,allcolors=blue,backref=page]{hyperref}
\usepackage{amsmath, amssymb, amsthm}
\usepackage{mathrsfs}
\usepackage{mathtools}
\usepackage[noabbrev,capitalize,nameinlink]{cleveref}
\crefname{equation}{}{}
\usepackage{fullpage}
\usepackage[noadjust]{cite}
\usepackage{graphics}
\usepackage{pifont}
\usepackage{tikz}
\usepackage{bbm}
\usepackage[T1]{fontenc}

\usetikzlibrary{arrows.meta}

\usepackage{environ}
\usepackage{framed}
\usepackage{url}
\usepackage[linesnumbered,ruled,vlined]{algorithm2e}
\usepackage[noend]{algpseudocode}
\usepackage[labelfont=bf]{caption}
\usepackage{cite}
\usepackage{framed}
\usepackage[framemethod=tikz]{mdframed}
\usepackage{appendix}
\usepackage{graphicx}
\usepackage[textsize=tiny]{todonotes}
\usepackage{tcolorbox}
\usepackage{enumerate}
\usepackage{verbatim}
\allowdisplaybreaks[1]

\apptocmd{\sloppy}{\hbadness 10000\relax}{}{} 

\crefname{algocf}{Algorithm}{Algorithms}

\crefname{equation}{}{} 
\crefname{conjecture}{Conjecture}{Conjectures} 
\AtBeginEnvironment{appendices}{\crefalias{section}{appendix}} 

\usepackage[color,final]{showkeys} 

\colorlet{refkey}{orange!20}
\colorlet{labelkey}{blue!30}

\crefname{algocf}{Algorithm}{Algorithms}

\numberwithin{equation}{section}
\newtheorem{theorem}{Theorem}[section]

\newtheorem{lemma}[theorem]{Lemma}

\crefname{claim}{Claim}{Claims}

\newtheorem{corollary}[theorem]{Corollary}
\newtheorem{conjecture}[theorem]{Conjecture}
\newtheorem*{question*}{Question}

\theoremstyle{definition}
\newtheorem{definition}[theorem]{Definition}

\newtheorem*{definition*}{Definition}

\theoremstyle{remark}
\newtheorem*{remark}{Remark}

\newcommand{\mb}{\mathbb}

\newcommand{\mbm}{\mathbbm}
\newcommand{\mc}{\mathcal}

\newcommand{\on}{\operatorname}
\newcommand{\wh}{\widehat}

\newcommand{\lam}{\lambda}
\newcommand{\E}{\mb{E}}
\newcommand{\var}{\on{Var}}

\allowdisplaybreaks

\title{Approximate counting and sampling\\ via local central limit theorems}

\author[Jain]{Vishesh Jain}
\address{Department of Statistics, Stanford University}
\email{visheshj@stanford.edu}

\author[Perkins]{Will Perkins}
\address{Department of Mathematics, Statistics, and Computer Science, University of Illinois at Chicago}
\email{math@willperkins.org}

\author[Sah]{Ashwin Sah}
\author[Sawhney]{Mehtaab Sawhney}

\address{Department of Mathematics, Massachusetts Institute of Technology}
\email{\{asah,msawhney\}@mit.edu}

\thanks{Perkins was supported in part by NSF grants DMS-1847451 and CCF-1934915.  Sah and Sawhney were supported by NSF Graduate Research Fellowship Program DGE-1745302.}

\begin{document}

\begin{abstract}
We give an FPTAS for computing the number of matchings of size $k$ in a graph $G$ of maximum degree $\Delta$ on $n$ vertices, for all $k \le (1-\delta)m^*(G)$, where $\delta>0$ is fixed and $m^*(G)$ is the matching number of $G$, and an FPTAS for the number of independent sets of size $k \le (1-\delta) \alpha_c(\Delta) n$, where $\alpha_c(\Delta)$ is the NP-hardness threshold for this problem. We also provide quasi-linear time randomized algorithms to approximately sample from the uniform distribution on matchings of size $k \leq (1-\delta)m^*(G)$ and independent sets of size $k \leq (1-\delta)\alpha_c(\Delta)n$. 

Our results are based on a new framework for exploiting local central limit theorems as an algorithmic tool. We use a combination of Fourier inversion, probabilistic estimates, and the deterministic approximation of partition functions at complex activities to extract approximations of the coefficients of the partition function. For our results for independent sets, we prove a new local central limit theorem for the hard-core model that applies to all fugacities below $\lam_c(\Delta)$, the uniqueness threshold on the infinite $\Delta$-regular tree.
\end{abstract}

\maketitle

\section{Introduction}\label{sec:introduction}

Counting matchings and independent sets in graphs are central problems in the study of exact and approximate counting algorithms.  Exact counting of the total number of matchings of a graph, the number of perfect matchings, the number of matchings of a given size, the number of independent sets, and the number of independent sets of a given size are all \#P-hard problems~\cite{valiant1979complexity}, even for many restricted classes of input graphs (bipartite graphs, graphs of bounded degree).  A singular exception is the classical algorithm of Kasteleyn for counting the number of perfect matchings of a planar graph~\cite{kasteleyn1967graph}.

Turning to approximate counting, a stark difference emerges between matchings and independent sets. The landmark work of Jerrum and Sinclair gave an FPRAS (fully polynomial-time randomized approximation scheme) for counting (weighted) matchings in general graphs as well as counting matchings of any given size bounded away from the maximum matching~\cite{jerrum1989approximating}. For the special case of bipartite graphs, Jerrum, Sinclair, and Vigoda \cite{jerrum2004polynomial} gave an FPRAS for the number of matchings of any given size, including perfect matchings. Recent work of Alimohammadi, Anari, Shiragur, and Vuong~ \cite{alimohammadi2021fractionally} provides an FPRAS for the number of matchings of any given size in planar graphs. 
On the other hand, for counting (weighted) independent sets and independent sets of a given size, there is a threshold (in terms of degree, weighing factor, or density) above which the approximation problems are NP-hard and below which efficient approximation algorithms exist~\cite{weitz2006counting,DP21, anari2020spectral}. We mention that approximating the number of perfect matchings in general graphs is an outstanding open problem (see \cite{vstefankovivc2018counting}). 

Randomization has played a crucial role in the aforementioned algorithmic results, and especially for matchings, there is a wide gap between what is known to be achievable deterministically and with randomness. For graphs of maximum degree $\Delta$, Bayati, Gamarnik, Katz, Nair, and Tetali~\cite{bayati2007simple} gave an FPTAS (fully polynomial-time apporoximation scheme) for the number of matchings  and for weighted matchings with a bounded weighing factor; the running time of their algorithm is polynomial in $n$ and $1/\epsilon$, where $\epsilon$ is the desired accuracy and where the exponent of the polynomial depends on $\Delta$ and the weighing factor. However, it is not known, for instance, how to deterministically approximate the number of near-maximum matchings in bounded degree graphs. Our first main result addresses this by achieving, for bounded degree graphs, what was previously only known to be possible with randomness: we provide an FPTAS for the number of matchings of any given size bounded away from the maximum matching.  Let $m_k(G)$ be the number of matchings of $G$ of size $k$ and let $m^*(G)$ be the size of the maximum matching of $G$.

\begin{theorem}
\label{thm:matching-fptas}
Let $\Delta \geq 3$ and $\delta \in (0,1)$. There exists a deterministic algorithm which, on input a graph $G = (V,E)$ on $n$ vertices of maximum degree at most $\Delta$, an integer $1\leq k \leq (1-\delta)m^*(G)$, and an error parameter $\epsilon \in (0,1)$, outputs an $\epsilon$-relative approximation to $m_k(G)$ in time $(n/\epsilon)^{O_{\delta, \Delta}(1)}$.
\end{theorem}
Here, by an $\epsilon$-relative approximation to $m_k(G)$, we mean that the output $A$ satisfies $e^{-\epsilon}m_k(G) \leq A\leq e^{\epsilon}m_k(G)$. 

We next prove the corresponding result for independent sets. Recall that the hard-core model on a graph $G = (V,E)$ at fugacity $\lambda \in \mb{R}_{\ge 0}$ is the probability distribution on $\mc{I}(G)$, the independent sets of $G$, defined by
\[\mu_{G,\lambda}(I) = \frac{\lambda^{|I|}}{Z_G(\lambda)},\]
where $Z_G(\lambda) = \sum_{I \in \mc{I}(G)}\lambda^{|I|} $ is the independence polynomial of $G$. 
For $\Delta \ge 3$, let 
 $$\lambda_c(\Delta) = \frac{(\Delta-1)^{\Delta-1}}{(\Delta-2)^{\Delta}};$$
 this is the uniqueness threshold for the hard-core model on the infinite $\Delta$-regular tree.  For $ 0\leq \lam < \lam_c(\Delta)$,  Weitz  gave an FTPAS for $Z_G(\lambda)$ on the class of graphs of maximum degree $\Delta$~\cite{weitz2006counting}.  Sly~\cite{sly2010computational}, Sly and Sun~\cite{sly2014counting}, and Galanis, {\v{S}}tefankovi{\v{c}}, and Vigoda~\cite{galanis2016inapproximability} complemented this by showing that for $\lambda > \lambda_c(\Delta)$, no FPRAS for $Z_G(\lam)$ exists unless $\on{NP}=\on{RP}$. 
 
 Recently, Davies and Perkins \cite{DP21} showed an analogous threshold for counting independent sets of a given size in bounded degree graphs. Let 
\[ \alpha_c(\Delta) = \frac{\lam_c(\Delta)}{1+(\Delta+1)\lam_c(\Delta)} = \frac{(\Delta-1)^{\Delta-1}}{(\Delta-2)^{\Delta}+(\Delta+1)(\Delta-1)^{\Delta-1}};\]
this is the occupancy fraction (i.e.~the expected density of an independent set) for the hard-core model on the clique on $\Delta + 1$ vertices at the critical fugacity $\lambda_c(\Delta)$.
They showed that for $\alpha < \alpha_c(\Delta)$, there is an FPRAS for $i_k(G)$ (the number of independent sets in $G$ of size $k$) for any  $G$ of maximum degree $\Delta$ on $n$ vertices and any $k \le \alpha n$; conversely, no FPRAS exists for $k \ge \alpha n$ for $\alpha > \alpha_c(\Delta)$ unless $\on{NP} = \on{RP}$.  The algorithm of~\cite{DP21} uses randomness in an essential way and the authors conjectured the existence of an FPTAS for $i_k(G)$ for $k\leq \alpha n$ where $\alpha <\alpha_c(\Delta)$~\cite[Conjecture~1]{DP21}.  We prove this conjecture. 
\begin{theorem}
\label{thm:independent-fptas}
Let $\Delta \geq 3$ and $\delta \in (0,1)$. There exists a deterministic algorithm which, on input a graph $G = (V,E)$ on $n$ vertices of maximum degree at most $\Delta$, an integer $1\leq k \leq (1-\delta)n\alpha_{c}(\Delta)$, and an error parameter $\epsilon \in (0,1)$, outputs an $\epsilon$-relative approximation to $i_k(G)$ in time $(n/\epsilon)^{O_{\delta, \Delta}(1)}$.
\end{theorem}

We remark that a deterministic algorithm for approximating $m_k$ and $i_k$ via the cluster expansion is implicit in~\cite{davies2020proof}, but only for $k $ much smaller than the bounds above.

Next, we turn to the problem of uniformly sampling matchings and independent sets of a given size. For fixed $\Delta \geq 3$, $\delta \in (0,1)$,  Davies and Perkins \cite{DP21} gave an algorithm for $\epsilon$-approximately sampling (i.e.~within $\epsilon$ in total variation distance) from the uniform distribution on $\mc{I}_k(G)$ (the independent sets of size $k$ in $G$) for a graph $G = (V,E)$ on $n$ vertices with maximum degree at most $\Delta$, for all $1\leq k \leq (1-\delta)n\alpha_c(\Delta)$. The running time of their algorithm is $\tilde{O}_{\delta, \Delta}(n^3)$, where the $\tilde{O}$ conceals polylogarithmic factors in $n$ and $1/\epsilon$. For the more restricted range $1\leq k \leq (1-\delta)n/(2(\Delta+1))$, it was already shown by Bubley and Dyer \cite{bubley1997path} that the down-up walk on $\mc{I}_k(G)$ has $\epsilon$-mixing time $O_{\delta}(n\log(n/\epsilon))$, which is optimal up to constants (see also~\cite{alev2020improved} which gave fast mixing for a larger range of $k$ in graphs satisfying a spectral condition).  The down-up walk is the following Markov chain on $\mc{I}_k(G)$: at each  step, given the current independent set $I \in \mc{I}_k(G)$, choose a uniformly random vertex $v \in I$ and a uniformly random vertex $w \in V$. Let $I' = (I \setminus v) \cup w$. If $I' \in \mc{I}_k(G)$, then move to $I'$; else, stay at $I$. 

For matchings, a polynomial time algorithm for approximately sampling from the uniform distribution on matchings of size $k$, for all $1\leq k \leq (1-\delta)m^*(G)$, is present in the work of Jerrum and Sinclair \cite{jerrum1989approximating}. The running time of their algorithm scales at least as $n^{7/2}$. In the case of graphs of maximum degree at most $\Delta$, a recent result of Chen, Liu, and Vigoda \cite{CLV20}  combined with a rejection sampling procedure (and \cref{lem:lambda-determine} below) provides an algorithm for this task running in time $\tilde{O}_{\delta, \Delta}(n^{3/2})$. 

It was conjectured in \cite[Conjecture~2]{DP21} that the down-up walk on $\mc{I}_k(G)$ mixes rapidly for all $1\leq k \leq (1-\delta)n\alpha_c(\Delta)$ on all graphs on $n$ vertices of maximum degree at most $\Delta$. A stronger conjecture is that the $\epsilon$-mixing time of the chain for this range of $k$ is $O_{\delta, \Delta}(n\log(n/\epsilon))$. While not resolving this specific conjecture, our next main result provides an approximate sampling algorithm for $1\leq k \leq (1-\delta)n\alpha_c(\Delta)$ (and $1\leq k \leq (1-\delta)m^*(G)$ in the case of matchings) running in quasi-linear time, which matches (up to a small polylogarithmic factor) the conjectured mixing time of the down-up walk. 
\begin{theorem}
\label{thm:faster-sampling}
Let $\Delta \geq 3$ and $\delta \in (0,1)$. There is a randomized algorithm which, on input a graph $G = (V,E)$ on $n$ vertices of maximum degree at most $\Delta$, an integer $1 \leq k \leq (1-\delta)n\alpha_c(\Delta)$, and an error parameter $\epsilon \in (0,1)$ outputs a random independent set $I \in \mc{I}_k(G)$ such that the total variation distance between the distribution of $I$ and the uniform distribution on $\mc{I}_k(G)$ is at most $\epsilon$. The running time of the algorithm is $O_{\delta, \Delta}(n\log(n/\epsilon)(\log n)^3 + n\log(n/\epsilon)\log n\log(1/\epsilon)^{3/2})$. 

There is also a randomized algorithm with the same guarantee and running time for matchings of size $k$ for all $1\leq k \leq (1-\delta)m^*(G)$. 
\end{theorem}

A natural extension of the unresolved conjecture of~\cite{DP21} is that the down--up walk for matchings of size $k$ is rapidly mixing in the setting of Theorem~\ref{thm:faster-sampling}. 
\begin{conjecture}
The down--up walk for matchings of size $k$ mixes in time $O_{\Delta,\delta}(n \log(n/\epsilon))$ for graphs $G$ of maximum degree $\Delta$ and  $1\leq k \leq (1-\delta)m^*(G)$.
\end{conjecture}

\subsection{Local central limit theorems} In previous works on approximately counting matchings \cite{jerrum1989approximating} and independent sets \cite{DP21} of a given size, a common approach is followed: sample a matching or independent set from the monomer-dimer model (i.e.~hard-core model on the line graph) or hard-core model where the fugacity $\lam$ is chosen so that the average size is close to the desired size $k$. If the fugacity $\lambda$ is such that sampling from the corresponding monomer-dimer or hard-core model can be done efficiently, and if one can further show that the probability of obtaining a matching or independent set of size exactly $k$ is only polynomially small, then na\"ive rejection sampling gives an efficient sampling algorithm (which can be converted into an FPRAS for $m_k(G)$ or $i_k(G)$ by standard self-reducibility techniques). In the proofs of \cref{thm:matching-fptas,thm:independent-fptas}, we show how to implement a version of this idea deterministically.   

Recall that a sequence of random variables $X_n$ with mean $\mu_n$ and variance $\sigma^2_n$ is said to satisfy a central limit theorem (CLT) if for all $a \leq b \in \mb{R}$,
\[\mb{P}\left[a \leq \frac{X_n - \mu_n}{\sigma_n} \leq b\right] = \frac{1}{\sqrt{2\pi}}\int_{a}^{b}e^{-x^2/2}dx + o_n(1).\]
In particular, central limit theorems provide control on the probability that $X_n$ lies in an interval of length $\Theta(\sigma_n)$. A much more precise notion is that of a local central limit theorem (LCLT). We say that a sequence of integer-valued random variables $X_n$ with mean $\mu_n$ and variance $\sigma^2_n$ satisfies an LCLT if for all integers $k$, 
\[ \mathbb{P}[X_n=k] = \frac{1}{\sqrt{2 \pi} \sigma_n} e^{ -(k-\mu_n)^2/(2 \sigma_n^2)}  + o_n \left( \sigma_n^{-1}  \right).\]
Returning to the discussion in the previous paragraph, suppose we could deterministically find a fugacity $\lambda$ such that:
\begin{enumerate}[(a)]
\item The expected size (to the nearest integer) of a matching or independent set drawn from the  monomer-dimer or hard-core model is $k$.
\item There is an FPTAS for the the partition function, expectation, and variance of the monomer-dimer or hard-core model at $\lambda$.
\item The size of a matching or independent set drawn from the monomer-dimer or hard-core model at $\lambda$ satisfies an LCLT.
\end{enumerate}
Then, from (a) and (c), we have that (as long as $\sigma_n = \omega_n(1)$)
\[\frac{m_k(G_n)\lambda^k}{\sum_{j=0}^{n}m_j(G_n)\lambda^j}=\mb{P}[Y_n = k] = (1+o_n(1))\frac{1}{\sqrt{2\pi}\sigma_n}e^{-(k-\mu_n)^2/(2\sigma_n^2)},\]
and similarly for independent sets. 
Together with (b), this immediately gives a deterministic algorithm for approximating $m_k(G_n)$ or $i_k(G_n)$ to within a factor of $(1+\epsilon)(1+o_n(1))$ in time which is polynomial in $n$ and $1/\epsilon$.

For the range of $k$ covered by \cref{thm:matching-fptas,thm:independent-fptas}, a fugacity $\lambda$ satisfying (a) and (b) does indeed exist and can be found deterministically using (b) along with a binary search procedure. In particular, for the hard-core model, such a $\lambda$ satisfies $\lambda < \lambda_c(\Delta)$. Moreover, by the Heilman--Lieb theorem~\cite{heilmann1972theory}, the roots of the partition function of the monomer-dimer model are restricted to the negative real line, from which it follows that the size of a random matching is distributed as the sum of independent Bernoulli random variables and thus, for $\lambda$ not too small, satisfies an LCLT~\cite{godsil1981matching}. However, for the hard-core model, the corresponding LCLT for all $\lambda < \lambda_{c}(\Delta)$  was not known. In fact, even the much weaker statement that $\mb{P}_{\lambda}[|I| = \lfloor \mu_\lambda \rfloor] = \Omega(\sigma_\lambda^{-1})$ (here, the probability $\mb{P}_{\lambda}$, the mean $\mu_{\lambda}$, and the variance $\sigma_{\lambda}^2$ are with respect to the hard-core model with fugacity $\lambda$) was unavailable for all $\lambda < \lambda_c(\Delta)$ and is precisely the content of \cite[Conjecture~3]{DP21}. A key step in our proof is the resolution of this conjecture in the much stronger form of an LCLT.

\begin{theorem}
\label{thm:independent-lclt}
Fix $\Delta\ge 3$ and $\delta \in (0,1)$. Then for any sequence of graphs $G_n$ on $n$ vertices of maximum degree $\Delta$, and any sequence $\lam_n \in \mb{R}^+$ so that $n \lam_n \to \infty$ and $\lam_n \le (1-\delta) \lam_c(\Delta)$, the size of a random independent set drawn from the hard-core model on $G_n$ at fugacity $\lam_n$ satisfies a local central limit theorem. 
\end{theorem}
In \cref{secLCLT}, we state and prove a quantitative version of this result (\cref{thm:independent-lclt2}).  A critical part in the proof of the LCLT is the existence of a suitable zero-free region (in the complex plane) for the independence polynomial, which has previously been exploited using Barvinok's method (cf.~\cite{barvinok2016combinatorics}) to devise an FPTAS for the independence polynomial evaluated at $\lambda$ in a certain complex region containing the interval $[0, \lambda_c(\Delta))$ 
on graphs of maximum degree $\Delta$ \cite{patel2017deterministic,PR19}.  

Given the LCLT for $\lambda$ satisfying (a) and (b), the above discussion immediately leads to a $(1\pm \epsilon)(1\pm o_n(1))$-factor approximation of $m_k(G)$ or $i_k(G)$ in time which is polynomial in $n$ and $1/\epsilon$, where the degree of the polynomial is allowed to depend on $\Delta$ and $\delta$. As such, this is only an EPTAS (efficient polynomial-time approximation scheme) since the finest approximation one can achieve with this method is $(1\pm o_n(1))$ (see \cref{sub:EPTAS} for a more detailed discussion). Nevertheless, we show that the \emph{proof} of the LCLT using Fourier inversion can be converted into an FPTAS, thereby providing a (perhaps surprising) connection between the deterministic approximation of the matching polynomial or independence polynomial at \emph{complex} fugacities and the deterministic approximation of suitable coefficients of the matching polynomial or independence polynomial. We note that the computational complexity of evaluating partition functions at complex parameters has received much recent attention \cite{harvey2018computing,PR19,bezakova2019inapproximability,liu2019ising,buys2021lee}. 

In a nutshell, the idea is the following: given a graph $G$ with maximum degree $\Delta$, an error parameter $\epsilon \in (0,1)$, and $k$ as in \cref{thm:independent-fptas}, consider $\lambda < \lambda_{c}(\Delta)$ satisfying (a) and (b). Let $Y=|I|$ denote the size of an independent set drawn from the hard-core model on $G$ at fugacity $\lambda$, let $\mu$ and $\sigma^2$ denote the mean and variance of $Y$, and let $X = (Y-\mu)/\sigma$. By the Fourier inversion formula for lattices, for all $x \in \sigma^{-1}\cdot \mb{Z} - \sigma^{-1}\mu$, 
\[\mb{P}[X = x] = \frac{1}{2\pi \sigma}\int_{-\pi \sigma}^{\pi \sigma} \mb{E}[e^{itX}]e^{-itx}dt,\]
where the probability and expectations are with respect to the hard-core model with fugacity $\lambda$.

In order to approximate $\mb{P}[X = x]$ we consider $\mb{E}[e^{itX}]$. First, observe that
\[\mb{E}[e^{itX}] = e^{-it\mu/\sigma}\cdot \frac{Z_G(\lambda e^{it/\sigma})}{Z_G(\lambda)}.\]
If $t/\sigma$ is sufficiently small so that $\lambda e^{-it/\sigma}$ is in the zero-free region of $Z_G(\lambda)$ (viewed as a univariate polynomial of a single complex variable), then there is an FPTAS for $Z_G(\lambda e^{it/\sigma})$ by Barvinok's method \cite{barvinok2016combinatorics, PR19, patel2017deterministic}. However this method does not handle all $t$. But over the range of $t$ sufficiently large, our proof of the LCLT shows that the integral is bounded by $(\epsilon/2)\mb{P}[X=x]$. It turns out that these two regimes overlap. Therefore using a Riemann sum approximation for the small regime gives an FPTAS for $i_k(G)$.

Towards the proof of \cref{thm:faster-sampling}, given $k \leq (1-\delta)\alpha_c(\Delta)n$, we provide in \cref{lem:lambda-determine} a $\tilde{O}_{\delta, \Delta}(n)$ time randomized algorithm for finding $\lambda$ satisfying $|\mb{E}_{\lambda}Y - k| \leq \sqrt{\on{Var}_{\lambda}Y}$. As mentioned earlier, this can be combined with the $\tilde{O}_{\Delta, \delta}(n)$ mixing of the Glauber dynamics for the hard-core model at fugacity $\lambda$ \cite{CLV20} and rejection sampling to approximately sample from the uniform distribution on $\mc{I}_k(G)$ in time $\tilde{O}_{\Delta, \delta}(n^{3/2})$, since the acceptance probability is $\tilde{\Omega}_{\Delta, \delta}(1/\sqrt{\on{Var}_{\lambda}|I|})$. The main idea underlying our algorithm is that the one may instead perform rejection sampling with the base distribution (effectively) being the hard-core distribution conditioned on $Y\equiv k \bmod p$, with $p = \tilde{\Theta}(\sqrt{\on{Var}_{\lambda}Y})$, while still ensuring that samples from the base distribution can be obtained in time $\tilde{O}(n)$. Given this, the assertion of \cref{thm:faster-sampling} follows since by the LCLT, the acceptance probability is now $\mb{P}_{\lambda}[Y=k]/\mb{P}_{\lambda}[Y\equiv k \bmod p] = \tilde{\Omega}_{\Delta, \delta}(1)$. Thus, the key step is to show that samples from the hard-core distribution conditioned on $Y\equiv k \bmod p$ may still be obtained in time $\tilde{O}(n)$. For this, we use a multi-stage view of sampling from the hard-core model, motivated directly by the proof of the LCLT.

\subsection{Additional results}
In \cref{thm:independent-fptas,thm:matching-fptas}, we gave an FPTAS for $i_k(G)$ and $m_k(G)$ with running times of the form $(n/\epsilon)^{O_{\Delta, \delta}(1)}$. Now, we provide substantially faster randomized algorithms for the same problems. 

In \cite{DP21}, an FPRAS (fully polynomial-time randomized approximation scheme) for $i_k(G)$ was given, for all $1\leq k \leq (1-\delta)n\alpha_c(\Delta)$, with running time $\tilde{O}_{\Delta, \delta}(n^6\epsilon^{-2})$. Combining this algorithm with our near-optimal sampling algorithm improves the running time to $\tilde{O}_{\Delta, \delta}(n^{4}\epsilon^{-2})$. In contrast, it is known that there is an FPRAS for the independence polynomial $Z_G(\lambda)$, for all $0 \leq \lambda \leq (1-\delta)\lambda_c(\Delta)$, with running time $\tilde{O}(n^2\epsilon^{-2})$ (cf.~\cite{vstefankovivc2009adaptive}). Our next result provides an FPRAS for $i_k(G)$ and $m_k(G)$ whose running time exceeds the best-known running time for approximating the independence polynomial or matching polynomial only by a lower order term.   

\begin{theorem}
\label{thm:faster-fpras}
Let $\Delta \geq 3$ and $\delta \in (0,1)$. There is a randomized algorithm which, on input a graph $G = (V,E)$ on $n$ vertices of maximum degree at most $\Delta$, an integer $1\leq k \leq (1-\delta)n\alpha_c(\Delta)$ and an error parameter $\epsilon \in (0,1)$ outputs an $\epsilon$-relative approximation to $i_k(G)$ with probability $3/4$ in time
\[T + O_{\Delta, \delta}(n^{3/2}\log n\log (n/\epsilon)\epsilon^{-2}),\]
where $T$ is the time to find an $\epsilon/2$-relative approximation to $Z_G(\lambda)$ for a given $\lambda \in [0, (1-\delta)\lambda_c(\Delta)]$.

Moreover, there exists a constant $C_{\Delta, \delta} > 0$ such that the same conclusion holds for $m_k(G)$ for all $1 \leq k \leq (1-\delta)m^*(G)$, where $T$ is the time to find an $\epsilon/2$-relative approximation to the matching polynomial $Z_G(\lambda)$ for a given $\lambda \in [0, C_{\Delta, \delta}]$. 
\end{theorem}
\begin{remark}
For $1 \leq k \leq c_{\Delta}\sqrt{n}$ (where $c_{\Delta}$ is a sufficiently small constant), the running time may be improved to
\[O(k\log{n}\log\log{n}) + O_{\Delta}(k\epsilon^{-2}\log{n}).\]
Moreover, the term $n^{3/2}$ in \cref{thm:faster-fpras} may be replaced by $\tilde{O}(n)$ by using similar ideas as in the proof of \cref{thm:faster-sampling}. However, since the current bounds on $T$ are $\Omega(n^2)$, we have not pursued this improvement.
\end{remark}

Finally, as an intermediate step in obtaining deterministic approximate counting algorithms, we need to deterministically approximate the mean and variance of the size of an independent set or matching drawn from the hard-core or monomer-dimer model respectively.  While such algorithms are obtainable by applying algorithms based on the correlation decay method of Weitz~\cite{weitz2006counting,bayati2007simple} to approximate marginals and joint marginals, we instead provide faster deterministic algorithms by adapting the method of Barvinok~\cite{barvinok2016combinatorics} and Patel--Regts~\cite{patel2017deterministic} to approximate the $k$th cumulant of the size of the random independent set or matching.

 The $k$th cumulant of a random variable $Y$ is defined in terms of the coefficients of the cumulant generating function $K_Y(t) =  \log \E e^{tY}$ (when this expectation exists in a neighborhood of $0$).  In particular, the $k$th cumulant is 
\[ \kappa_k(Y) = K^{(k)}_Y(0) \,. \]
The first and second cumulants are the mean and variance respectively. 
\begin{theorem}
\label{thm:linear-time-cumulants}
Fix $\Delta \ge 3$, $k\ge 1$, and $\delta \in(0,1)$.  There is an algorithm which, on input a graph $G = (V,E)$ on $n$ vertices of maximum degree at most $\Delta$, $0 < \lambda \le (1- \delta) \lambda_c(\Delta)$, and an error parameter $\epsilon \in (0,1)$ outputs an $\epsilon \lambda n$ additive approximation to $\kappa_k(Y)$, where $Y$ is the size of an independent set drawn from the hard-core model on $G$ at fugacity $\lambda$. The algorithm runs in time $O_{\Delta, \delta, k}(n (1/\epsilon)^{O_{\Delta,\delta}(1)})$.  In particular,  this provides an FPTAS for $\E_{\lambda} Y$ and $\var_{\lambda}Y$ running in time linear in $n$.

For claw-free graphs (hence for the size of a random matching $Y$ drawn from the monomer-dimer model), the same holds for all bounded $\lambda > 0$. 

\end{theorem}

\subsection{Outline}
We prove \cref{thm:matching-fptas,thm:independent-fptas} in essentially the same way, given as an input the existence of a zero-free region for the corresponding partition function in the complex plane.  For matchings, this is provided by the classical Heilmann--Lieb theorem~\cite{heilmann1972theory}, while for independent sets this is provided by the theorems of Shearer~\cite{shearer1985problem} and Peters and Regts~\cite{PR19}. To avoid excessive repetition (and to gain a slight amount of generality) we work with the larger class of independent sets in claw-free graphs (instead of matchings, which are independent sets in line graphs). The generalization of the Heilmann--Lieb theorem to claw-free graphs is due to Chudnovsky and Seymour~\cite{chudnovsky2007roots}. 

We record these zero-freeness results in \cref{sec:prelims}, along with some results from the geometry of polynomials on the consequences of zero-freeness, namely a central limit theorem of Michelen and Sahasradbudhe~\cite{MS19} and a deterministic approximation algorithm for $Z_G(\lam)$ (with $\lam$ possibly complex) due to Barvinok~\cite{barvinok2016combinatorics} and Patel and Regts~\cite{patel2017deterministic}. We also record a recent result of Chen, Liu, and Vigoda~\cite{CLV20} on the optimal mixing of Glauber dynamics for bounded marginal spin-systems on bounded degree graphs which is used in our randomized algorithms.  

In \cref{secLCLT}, we prove the local central limit theorem for the hard-core model (\cref{thm:independent-lclt}).  Our proof uses both the aforementioned central limit theorem and a variant of the technique of Dobrushin and Tirozzi~\cite{dobrushin1977central} who proved local central limit theorems for spin models on the integer lattice $\mathbb{Z}^d$.

In \cref{sec:DetAlgorithms}, we prove our deterministic algorithmic results \cref{thm:matching-fptas,thm:independent-fptas}. 

In \cref{sec:RandAlgorithms}, we prove our randomized algorithmic results \cref{thm:faster-fpras,thm:faster-sampling}. 

In \cref{secCluster}, we provide a proof of \cref{thm:independent-lclt} when $\lam$ is sufficiently small as a function of $\Delta$ using the cluster expansion, a tool from classical statistical physics. While not necessary for the main results, we include this since the proof is simpler and perhaps more intuitive. 

Finally, in \cref{secCumulants}, we prove \cref{thm:linear-time-cumulants}.

\subsection{Notation}\label{sub:notation}
Throughout, we reserve the random variable $Y$ for the size of a random independent set drawn from the hard-core model on a graph $G$.  We use subscripts to indicate the fugacity, e.g. $\E_\lam Y$ and $\var_\lam Y$.  We let $\alpha_G(\lam)$ denote the occupancy fraction of the hard-core model on $G$ at fugacity $\lam$; that is, $\alpha_G(\lam) = \frac{\E_{\lam} Y }{|V(G)| }$.     We let $\mc{Z}$ be a standard normal random variable and $\mc{N}(x) = e^{-x^2/2}/\sqrt{2\pi}$ denote its density. 

Dependence of various constants on input parameters will often be important; we will write $f(n) = O_\Delta(1)$, for instance, to mean that $f$ is bounded by a constant that depends only on $\Delta$.

\section{Preliminaries}\label{sec:prelims}

\begin{definition}
\label{def:approximation}
Let $z_1, z_2 \in \mb{C}$. 
We say that $z_1$ is a $\delta$-additive, $\epsilon$-relative approximation of $z_2$ if
$z_1 = re^{i\theta}z_2 + z_3$
for some $e^{-\epsilon} \leq r \leq e^{\epsilon}$, $\theta \in \mb{R}$, $|\theta| \leq \epsilon$ and $z_3 \in \mb{C}, |z_3| \leq \delta$.
When $\delta = 0$, we simply say that $z_1$ is an $\epsilon$-relative approximation of $z_2$. 
\end{definition}

The following theorem combines results of Shearer~\cite{shearer1985problem} on the non-vanishing of the independence polynomial in a complex disk, Peters and Regts~\cite{PR19} on the non-vanishing of the independence polynomial in a complex neighborhood of $(0,\lambda_c(\Delta))$, and Chudnovsky and Seymour~\cite{chudnovsky2007roots} on the real-rootedness of the independence polynomial of claw-free graphs (extending the Heilmann--Lieb theorem on the real-rootedness of the matching polynomial~\cite{heilmann1972theory}).
\begin{theorem}[\cite{shearer1985problem,PR19,chudnovsky2007roots,heilmann1972theory}]
\label{thm:zero-free}
Let $\Delta \geq 3$ and $\delta \in (0,1)$. There exists $c_{\ref{thm:zero-free}} = c_{\delta, \Delta} > 0$ such that for any graph $G = (V,E)$ with maximum degree at most $\Delta$, the partition function $Z_G(\lambda)$ of the hard-core model does not vanish on the region
\[\mc{R}_{\delta, \Delta} := \{z \in \mb{C}: 0\leq \Re(z) \leq (1-\delta)\lambda_{c}(\Delta), |\Im(z)| \leq c_{\delta, \Delta}\} \cup \{z \in \mb{C} : |z| < (\Delta-1)^{\Delta-1}/\Delta^\Delta\}.\]
Moreover, if $G$ is claw-free, then all of the roots of $Z_G(\lambda)$ are on the negative real axis below $-e/(\Delta + 1)$. 
In particular, $Z_G(\lambda)$ does not vanish on the region
\[\mc{C}_{\Delta} := \mb{C} \setminus \{z \in \mb{C}: \Re(z) \leq -e/(\Delta+1)\}.\]
\end{theorem}

Next, we record two consequences of the above zero-free regions of the partition function. The first is an FPTAS for $Z_G(\lambda)$, provided that $\lambda \in \mb{C}$ lies in the zero-free region, and follows by using Barvinok's method \cite{barvinok2016combinatorics} and the work of Patel and Regts \cite{patel2017deterministic}.

\begin{theorem}[\cite{barvinok2016combinatorics, patel2017deterministic}]
\label{thm:patel-regts}
Let $\Delta \geq 3$ and $\delta \in (0,1)$. There exists a deterministic algorithm which, on input a graph $G = (V,E)$ on $n$ vertices with maximum degree at most $\Delta$, a (possibly complex) fugacity $\lambda \in \mc{R}_{\delta, \Delta}$, and an approximation parameter $\epsilon \in (0,1)$, returns an $\epsilon$-relative approximation of $Z_G(\lambda)$ in time $(n/\epsilon)^{O_{\delta, \Delta}(1)}$. 

Moreover, if $G$ is claw-free, then the same conclusion holds for any $\lambda \in \mc{C}_{\Delta}$ with running time $(n/\epsilon)^{O_{\Delta, |\lambda|}(1)}$. 
\end{theorem}

The second is a result of Michelen and Sahasradbudhe \cite{MS19} on converting zero-free regions of probability generating functions into central limit theorems. We note that in order to establish our results with slightly worse quantitative dependencies, earlier results of Lebowitz, Pittel, Ruelle, and Speer \cite{LPRS16} are sufficient. 
\begin{theorem}[{\cite[Theorem~1.2]{MS19}}]\label{thm:clt-gen}
Let $X$ be a random variable taking values in $\{0,1,\dots,n\}$ with mean $\mu$ and variance $\sigma^2$ and let
$f_X(z) = \sum_{k=0}^{n}\mb{P}[X=k]z^k$ denote its probability generating function. Let $\delta = \min_{\zeta}|\zeta-1|$, where $\zeta$ ranges over the (complex) roots of $f_X$. Then, 
\[\sup_{t\in\mb{R}}|\mb{P}[(X-\mu)/\sigma\le t]-\mb{P}[\mc{Z}\le t]| = O\bigg(\frac{\log n}{\delta\sigma}\bigg).\]
\end{theorem}

For our proof of the LCLT for the hard-core model, we will need the following simple lemma. The precise version we state here appears in work of Berkowitz \cite{Ber16} on quantitative local central limit theorems for triangle counts in $\mb{G}(n,p)$ and has been used, for instance, in further work  on local central limit theorems for general subgraph counts in random graphs~\cite{SS20}; we include the short proof for the reader's convenience.
\begin{lemma}[{\cite[Lemma~3]{Ber16}}]\label{lem:fourier-convert}
Let $X$ be a random variable supported on the lattice $\mc{L} = \alpha + \beta \mb{Z}$ and let $\mc{N}(x) = e^{-x^2/2}/\sqrt{2\pi}$ denote the density of the standard normal distribution. Then 
\[\sup_{x\in \mc{L}}|\beta\mc{N}(x) - \mb{P}[X=x]|\le \beta\int_{-\pi/\beta}^{\pi/\beta}\big|\mb{E}[e^{itX}]-\mb{E}[e^{it\mc{Z}}]\big|dt + e^{-\pi^2/(2\beta^2)}.\]
\end{lemma}
\begin{proof}
Let $\varphi_X(t) = \mb{E}e^{itX}$ and $\varphi(t) = \mb{E}e^{it\mc{Z}}$. By Fourier inversion on lattices and by Fourier inversion on $\mb{R}$, we have for $x\in\mc{L}$ that
\[\mb{P}[X=x] = \frac{\beta}{2\pi}\int_{-\pi/\beta}^{\pi/\beta}\varphi_X(t)e^{-itx}dt,\qquad\mc{N}(x) = \frac{1}{2\pi}\int_{-\infty}^\infty\varphi(t)e^{-itx}dt.\]
Therefore
\begin{align*}
|\beta\mc{N}(x)-\mb{P}[X=x]|&=\frac{\beta}{2\pi}\bigg|\int_{-\infty}^\infty\varphi(t)e^{-itx}dt-\int_{-\pi/\beta}^{\pi/\beta}\varphi_X(t)e^{-itx}dt\bigg|\\
&\le\frac{\beta}{2\pi}\int_{-\pi/\beta}^{\pi/\beta}|\varphi(t)-\varphi_X(t)|dt+\frac{\beta}{2\pi}\bigg|\int_{|t|>\pi/\beta}e^{-itx}\varphi(t)dt\bigg|\\
&\le\beta\int_{-\pi/\beta}^{\pi/\beta}|\varphi(t)-\varphi_X(t)|dt +e^{-\pi^2/(2\beta^2)},
\end{align*}
where we used a standard tail bound on Gaussian integrals in the last line. Taking the supremum over $x \in \mc{L}$ completes the proof. 
\end{proof}

For our randomized algorithms, we will make use of recent results of Chen, Liu, and Vigoda \cite{CLV20} establishing optimal mixing of the Glauber dynamics for bounded-degree spin systems. 
\begin{theorem}[\cite{CLV20}]
\label{thm:glauber}
Let $\Delta \geq 3$ and $\delta \in (0,1)$. For every graph $G = (V,E)$ on $n$ vertices with maximum degree at most $\Delta$ and for every $0 < \lambda \leq (1-\delta)\lambda_c(\Delta)$, the $\epsilon$-mixing time of the Glauber dynamics for the hard-core model on $G$ at fugacity $\lambda$ is $O_{\Delta, \delta}(n\log(n/\epsilon))$. 

Moreover, for any $C > 0$, the same conclusion holds for all line graphs $G = (V,E)$ on $n$ vertices with maximum degree at most $\Delta$ and any $0 < \lambda \leq C$, with the implicit constant depending on $\Delta$ and $C$. 
\end{theorem}

\section{Proof of the local central limit theorem}
\label{secLCLT}

The goal of this section is to prove the following quantitative version of \cref{thm:independent-lclt}.
\begin{theorem}
\label{thm:independent-lclt2}
Fix $\Delta\ge 3$ and $\delta>0$. Let $\lambda\in(0,\lambda_c(\Delta)-\delta)$. Given a graph $G$ on $n$ vertices with maximum degree at most $\Delta$, draw a random independent set $I$ from the hard core model $\mu$ with fugacity $\lambda$ and let $Y = |I|$. Let $\mu = \mb{E}_{\lambda}Y$ and $\sigma^2 = \on{Var}_{\lambda}Y$. Then 
\[\sup_{t\in \mb{Z}}|\sigma^{-1}\mc{N}((t-\mu)/\sigma) - \mb{P}[Y=t]|= O_{\Delta,\delta}\bigg(\on{min}\bigg(\frac{(\log n)^{5/2}}{\sigma^2},\frac{1}{\sigma^{2}} + \frac{\sigma^6(\log n)^2}{n}\bigg)\bigg).\]
Moreover, if $G$ is claw-free, then the same conclusion holds for any $\lambda \in (0,C)$, with the implicit constant depending on $\Delta$ and $C$. 
\end{theorem}

\begin{remark}
\cref{lem:var-bound} below shows that $\sigma^{2} = \Theta_{\Delta}(\lambda n)$ (or $\sigma^2 = \Theta_{\Delta, C}(\lambda n)$ in the case of claw-free graphs). Thus, for fixed $\Delta, \delta, C$ and for a sequence of $\lambda_n$ in the appropriate region for which $\lambda_n n \to \infty$, we see that the right hand goes to $0$, thereby establishing the qualitative claim of \cref{thm:independent-lclt}.    
\end{remark}

The proof of \cref{thm:independent-lclt2} requires a few steps. First, we use the zero-free regions given by \cref{thm:zero-free} in order to bound $\sigma$. 
\begin{lemma}\label{lem:var-bound}
Let $\Delta \geq 3$ and $\delta \in (0,1)$. For any graph $G = (V,E)$ on $n$ vertices with maximum degree at most $\Delta$ and for all $\lambda \in (0, \lambda_c(\Delta) - \delta)$, we have
\[\frac{ \lambda n}{ (\Delta+1) (1+\lambda)^{2+\Delta  }  }  \le \on{Var}_{\lambda}Y \le C_{\delta, \Delta} \lambda n \,.\]
Moreover, if $G$ is claw-free, then  for all $\lambda >0$,
\[\frac{ \lambda n}{ (\Delta+1) (1+\lambda)^{2+\Delta  }  }  \le \on{Var}_{\lambda}Y \le C_{ \Delta} \lambda n \,.\]
\end{lemma}
\begin{proof}
The lower bounds follow from~\cite[Lemma 9]{DP21}.
For the upper bounds, let $N = \alpha(G) \le n$ be the independence number of $G$ (which is also the degree of $Z_G(\lambda)$ as a polynomial in $\lambda$).  Since the constant term of the polynomial $Z_G(\lambda)$ is $1$, we can write 
\[ Z_G(\lambda) = \prod_{j=1}^N \left( 1- \lambda r_j  \right) \]
where $r_1, \dots , r_N$ are the inverses of the complex roots of $Z_G(\lambda)$.  Then
\begin{align*}
\on{Var}_{\lambda}Y &= \lambda^2 (\log Z_G(\lambda)) '' + \lambda (\log Z_G(\lambda))'  \\
&=  -\sum_{j=1}^N \left( \frac{\lambda^2 r_j^2 }{(1- \lambda r_j)^2}  + \frac{\lambda r_j}{ 1-\lambda r_j    }\right)  \\
&= - \lambda \sum_{j=1}^N \frac{r_j}{ (1-\lambda r_j)^2} \,.
\end{align*}
By \cref{thm:zero-free}, $|r_j| = O(\Delta)$ and  $|1/(1-\lambda r_j)^2 | = O_{\delta, \Delta}(1)$ when $\lambda \in (0,\lambda_c(\Delta)-\delta)$.  In the case of claw-free graphs, $|1/(1-\lambda r_j)^2 | = O_{\Delta}(1)$ for all $\lambda \ge 0$ by \cref{thm:zero-free}.   Combining these estimates completes the proof.
\end{proof}

Next, we control the low Fourier phases of our random variable. 

\begin{lemma}\label{lem:low-fourier}
Let $\Delta\ge 3$ and $\delta \in (0,1)$. For any graph $G = (V,E)$ on $n$ vertices with maximum degree at most $\Delta$ and for any $\lambda \in (0, \lambda_{c}(\Delta)-\delta)$, we have for all $t \in \mb{R}$ that
\[|\mb{E}[e^{itX}]-\mb{E}[e^{it\mc{Z}}]| = O_{\Delta,\delta}\bigg(\frac{(|t|(\log n)^{3/2}+\log n)}{\sigma}\bigg),\] 
where $X = (Y-\mu)/\sigma$.

Moreover, if $G$ is claw-free, then the same conclusion holds for all $\lambda > 0$ with the implicit constant depending only on $\Delta$.
\end{lemma}
\begin{proof}
We prove the statement for general graphs; the proof for claw-free graphs is completely analogous. By \cref{thm:clt-gen} combined with the zero-free region from \cref{thm:zero-free} we have for all $t \in \mb{R}$ that
\begin{equation}
\label{eq:clt}
|\mb{P}[X\le t]-\mb{P}[\mc{Z}\le t]| = O_{\Delta,\delta}\bigg(\frac{\log n}{\sigma}\bigg).
\end{equation}
Let $X'$ be $X$ convolved with a centered Gaussian of infinitesimally small variance so that $X'$ has density with respect to the Lebesgue measure on $\mb{R}$; it suffices to prove the claim for $X'$ and then pass to the limit. For this, we note that
\begin{align*}
\mb{E}[e^{itX'}] &= \int_{-\infty}^{\infty}e^{itz}p_{X'}(z)dz\\
&= \int_{|z|\le \tau}e^{itz}p_{X'}(z)dz \pm \mb{P}[|X'|\ge \tau]\\
&= \left[e^{itz}\bigg(\int_{-\tau}^{z}p_{X'}(z')dz'\bigg)\right]\bigg|_{z=-\tau}^{z=\tau} -\int_{-\tau}^{\tau}ite^{itz}\bigg(\int_{-\tau}^{z}p_{X'}(z')dz'\bigg)dz\pm \mb{P}[|X'|\ge \tau]\\
&= e^{it\tau}-\int_{-\tau}^{\tau}ite^{itz}\bigg(\int_{-\tau}^{z}p_{X'}(z')dz'\bigg)dz\pm \mb{P}[|X'|\ge \tau] - e^{it\tau} \mb{P}[|X'|\geq \tau]\\
&= e^{it\tau}-\int_{-\tau}^{\tau}ite^{itz}\mb{P}[X'\in[-\tau,z]] dz\pm \mb{P}[|X'|\ge \tau] - e^{it\tau}\mb{P}[|X'|\geq \tau]\\
&= e^{it\tau}-\int_{-\tau}^{\tau}ite^{itz}\mb{P}[X'\in[-\tau,z]] dz + e^{i\theta}\cdot O_{\Delta, \delta}\bigg(\frac{\log n}{\sigma} + e^{-\tau^2/4}\bigg),
\end{align*}
for some $\theta \in [0,2\pi)$, where in the last line, we have used \cref{eq:clt} along with a standard Gaussian tail bound. 
Applying the same calculation to $\mc{Z}$ instead of $X'$ and taking the difference, we find that
\begin{align*}
    |\mb{E}[e^{itX'}]-\mb{E}[e^{itZ}]| 
    &\leq |t|\left|\int_{-\tau}^{\tau}\left|\mb{P}[X'\in[-\tau,z]] - \mb{P}[\mc{Z} \in [-\tau, z]]\right| dz\right| +  O_{\Delta, \delta}\bigg(\frac{\log n}{\sigma} + e^{-\tau^2/4}\bigg)\\
    &\leq O_{\Delta, \delta}\left(\frac{(|\tau t| + 1)\log{n}}{\sigma} + e^{-\tau^2/4}\right).
\end{align*}
Since $\sigma \leq n$, setting $\tau = \sqrt{8\log{n}}$ gives the desired conclusion.
\end{proof}

Finally, we control the high Fourier phases of our random variable, following a similar strategy to that of Dobrushin and Tirozzi~\cite{dobrushin1977central}. We need the following elementary lemma.
\begin{lemma}\label{lem:well-separated}
Let $G = (V,E)$ be a graph on $n$ vertices with maximum degree at most $\Delta$. Then, there exists a subset $S \subseteq V$ of size $\Omega(n/\Delta^3)$ such that all vertices in $S$ are pairwise distance at least $4$ with respect to the graph distance. Moreover, there is an algorithm to find such a subset $S$ in time $O_{\Delta}(n)$. 
\end{lemma}
\begin{proof}
Let $v_1,\dots, v_n$ denote an arbitrary enumeration of the vertices. Initialize $S = \emptyset$. Consider the greedy algorithm which, at each time step, adds the first available vertex to the set $S$ and removes all vertices within distance $3$ of this vertex from consideration. The algorithm stops when there are no more available vertices. The algorithm runs in time $O_{\Delta}(n)$ and outputs a set $S$ such that any two vertices in $S$ have graph distance at least $4$. Moreover, since at each time, $O(\Delta^3)$ vertices are removed, it follows that $|S| = \Omega(n/\Delta^3)$. 
\end{proof}

\begin{lemma}\label{lem:high-fourier}
Let $\Delta \geq 3$ and $C\geq 1$. There exists $c_{\ref{lem:high-fourier}}(\Delta, C) > 0$ satisfying the following. For any graph $G = (V,E)$ on $n$ vertices with maximum degree at most $\Delta$ and for any $\lambda \in (0,C]$, let $X = (Y-\mu)/\sigma$. Then, for all $t \in [-\pi \sigma, \pi \sigma]$, we have
\[|\mb{E}[e^{-itX}]|\le \exp(-c_{\ref{lem:high-fourier}}(\Delta, C)\lambda n t^2/\sigma^2).\]
\end{lemma}
\begin{proof}
Rewriting the claim, it suffices to prove that for all $t\in \mb{R}$, $|t|\leq \pi$,
\[|\mb{E}[e^{-itY}]|\le \exp(-c_{\Delta, C} \lambda n t^2).\]
Let $S$ be a $4$-separated set of vertices of $G$ of size $s = \Omega(n/\Delta^3)$ coming from \cref{lem:well-separated}. 
Let $T$ be the set of vertices that are at distance at least $2$ from $S$ in $G$ and let $G[T]$ denote the graph on $T$ induced by $G$. Let $\nu$ denote the distribution on $\mc{I}(G[T])$ induced by the hard-core model. We sample $I$ by first sampling $J \sim \nu$ and then sampling from the conditional distribution (induced by the hard-core model and $J$) on $\mc{I}(G[v\cup N(v)])$ for each $v \in S$. The key observation is that these conditional distributions are mutually independent. In particular, given $J$, we can write
\[Y \stackrel{d.}{=} |J| + X_1+\cdots+X_s\]
where each $X_j$ is an independent random variable with support in $\{0,\ldots,\Delta+1\}$, a probability mass at $0$ of $\Omega_{\Delta,C}(1)$, and a probability mass of $\Omega_{\Delta,C}(\lambda)$ at $1$. Note that the implicit constant in $\Omega$ does not depend on the specific realisation of $J$. 

We claim that for all $|t|\leq \pi$ and all $j \in [s]$, for any realisation of $J$, 
\[|\mb{E}e^{-itX_j}|\le 1-c\lambda t^2\]
for some absolute $c = c_{\Delta,C} > 0$. Indeed, for any realisation of $J$, letting $X_j'$ denote an independent copy of $X_j$, we have
\begin{align*}
|\mb{E}e^{-itX_j}|^2 
&= \mb{E}e^{it(X_j-X_j')} \\
&= \mb{P}[X_j = X_j'] + \sum_{k=1}^{\Delta + 1}(\mb{P}[X_j - X_j' = k] + \mb{P}[X_j'-X_j = k])\cos(kt)
\\
&\leq \mb{P}[X_j = X_j'] + \sum_{k=2}^{\Delta + 1}(\mb{P}[X_j - X_j' = k] + \mb{P}[X_j'-X_j = k]) + 2\mb{P}[X_j - X_j' = 1]\cos(t)\\
&= 1 - 2\mb{P}[X_j - X_j' = 1] + 2\mb{P}[X_j - X_j' = 1]\cos(t)\\
&= 1 - 2\mb{P}[X_j - X_j'=1](1-\cos(t))\\
&\leq 1- \frac{1}{4}\mb{P}[X_j - X_j'=1]t^2 \leq 1 - \frac{1}{4}\mb{P}[X_j = 1]\mb{P}[X_j'=0]t^2 \leq 1 - c_{\Delta, C} \lambda t^2,
\end{align*}
as claimed. 

Finally, we have that for any $t \in [-\pi, \pi]$,
\begin{align*}
    |\mb{E}[e^{-itY}]| 
    &\leq \max_{J}|\mb{E}[e^{-itY} \mid J]|\\
    & = \max_{J}\prod_{j=1}^{s}|\mb{E}e^{-itX_j}|\\
    &\leq (1-c_{\Delta, C}\lambda t^2)^{s/2}\\
    &\leq \exp(-c'n\lambda t^2),
\end{align*}
for an appropriate $c'=c_{\Delta,C}' > 0$ and the result follows. 
\end{proof}

\subsection{Finishing the proof}\label{sub:finishing}

We now prove \cref{thm:independent-lclt2}.
\begin{proof}[{Proof of \cref{thm:independent-lclt2}}]
We prove the result for general graphs; the proof for claw-free graphs is essentially identical. Applying \cref{lem:fourier-convert} to $X = (Y-\mu)/\sigma \in\alpha+\beta\mb{Z}$, where $\alpha = -\mu/\sigma$ and $\beta = 1/\sigma$, and using  \cref{lem:var-bound,lem:low-fourier,lem:high-fourier} we see that for $\sigma \geq 2$, 
\begin{align*}
\sup_{x\in \mc{L}}|\beta\mc{N}(x) - \mb{P}[X=x]|&\le\frac{1}{\sigma}\int_{-\pi\sigma}^{\pi\sigma}\big|\mb{E}[e^{itX}]-\mb{E}[e^{it\mc{Z}}]\big|dt + e^{-\pi^2\sigma^2/2}\\
&\lesssim_{\Delta,\delta}\frac{1}{\sigma}\int_{-\pi\sigma}^{\pi\sigma}\min\bigg(\frac{|t|(\log n)^{3/2}+\log n}{\sigma},e^{-c_{\ref{lem:high-fourier}}\lambda nt^2/\sigma^2}+e^{-t^2/2}\bigg)dt + e^{-\pi^2\sigma^2/2} \\
&\lesssim_{\Delta,\delta}\frac{1}{\sigma}\int_{-\pi\sigma}^{\pi\sigma}\min\bigg(\frac{|t|(\log n)^{3/2}+\log n}{\sigma},e^{-c_{\Delta, \delta}t^2}+e^{-t^2/2}\bigg)dt + e^{-\pi^2\sigma^2/2}\\
&\lesssim_{\Delta,\delta}\frac{1}{\sigma}\int_{-C_{\Delta,\delta}\sqrt{\log \sigma}}^{C_{\Delta,\delta}\sqrt{\log \sigma}}\frac{|t|(\log n)^{3/2}+\log n}{\sigma}dt + \frac{1}{\sigma^2}\\
&\lesssim_{\Delta,\delta}\frac{(\log n)^{5/2}}{\sigma^2}.
\end{align*}

This gives the first term in the minimum in the statement of \cref{thm:independent-lclt2}. For the second term, we may assume that $1 \leq \sigma  < \log{n}$. Let $\lambda' = \lambda/(1+\lambda)$ and observe that the hard-core distribution at fugacity $\lambda$ is identical to the product distribution $\on{Ber}(\lambda')^{\otimes V}$ conditioned on the configuration being an independent set. Here, $\on{Ber}(\lambda')$ is the random variable which is $1$ (or occupied) with probability $\lambda'$ and $0$ (or unoccupied) otherwise. A trivial union bound argument shows that a random sample from $\on{Ber}(\lambda')^{\otimes V}$ is an independent set with probability at least $1 - \lambda'^2 \Delta n = 1 - O_{\Delta}(\lambda^2 \Delta n) = 1- O_{\Delta}(\sigma^4/n)$, where we have used \cref{lem:var-bound}. Therefore, the probability of any configuration under the hard-core model is within a factor of $1 \pm O_{\Delta}(\sigma^4/n)$ of the probability of the same configuration under $\on{Ber}(\lambda')^{\otimes V}$. 

Let $Y'$ denote the random variable counting the number of $1$s in a random sample from $\on{Ber}(\mu)^{\otimes V}$, and let $\mu'$ and $\sigma'$ denote the mean and standard deviation of $Y'$. Then, by the classical DeMoivre-Laplace central limit theorem (see \cite[Chap.~VII,~Theorem~(4)]{Pet75} for the quantitative version used here), we get that for any integer $k$,
\[\left|\mb{P}[Y'=k] - \frac{1}{\sigma'}\mc{N}\left(\frac{k-\mu'}{\sigma'}\right) \right| = O\left(\frac{1}{\sigma'^2}\right).\]
Moreover, from the comparison between the hard-core model and $\on{Ber}(\lambda')^{\otimes V}$ mentioned above, as well as the Chernoff bound for the Binomial distribution, we see that
\begin{align*}
    \mb{P}[Y = k] &= \mb{P}[Y' = k] \pm O_{\Delta}(\sigma^4/n)\\
    \mu &= \mu'(1 \pm O_{\Delta}(\sigma^4\log{n}/n))\\
    \sigma^2 &= \sigma'^2(1 \pm O_{\Delta}(\sigma^6(\log n)^2/n)).
\end{align*}
Substituting this in the above gives the desired conclusion.
\end{proof}

\section{Deterministic algorithms}\label{sec:DetAlgorithms}

For our deterministic algorithms, we will need the following preliminary lemma, which allows us to find a `good' fugacity at which to either apply the LCLT directly or algorithmize its proof. 

\begin{lemma}
\label{lem:find-fugacity}
Let $\Delta \geq 3$ and $\delta \in (0,1/2)$. There exists a constant $C_{\ref{lem:find-fugacity}} = C_{\ref{lem:find-fugacity}}(\Delta, \delta) \geq 1$ for which the following holds. 

For any $\alpha \leq (1-\delta)\alpha_{c}(\Delta)$, there exists a unique $\lambda_* < \lambda_{c}(\Delta)$ so that $\alpha_{K_{\Delta+1}}(\lambda_*) = \alpha$. Further, for any graph $G = (V,E)$ on $n$ vertices of maximum degree at most $\Delta$ and for any $1\leq k\leq \alpha n$, there exists an integer $t \in \{0,\dots, \lceil C_{\ref{lem:find-fugacity}} n \lambda_*  \rceil\}$ such that
\[|n\alpha_G(t/(C_{\ref{lem:find-fugacity}}n)) - k| \leq 1/2.\]
Moreover, if $G$ is claw-free, then for any $1\leq k \leq (1-\delta)i^*(G)$, there exists an integer $t \in \{0,\dots, \lceil n\cdot 8^{(\Delta+1)/\delta^2}\rceil\}$ such that the same conclusion holds.

In either case, such an integer $t$ can be found deterministically in time $n^{O_{\delta, \Delta}(1)}$. 
\end{lemma}

\begin{proof}
For general graphs, the existence of such an integer $t$ and constant $C_{\ref{lem:find-fugacity}}$ follows from \cite[Lemma~5]{DP21} (upon replacing \cite[Lemma~9]{DP21} by the optimal variance bound \cref{lem:var-bound}). For claw-free graphs, the existence of such an integer $t$ follows similarly as a simple consequence of the following observations: (i) $\alpha_G(\lambda)$ is monotonically increasing in $\lambda$; (ii) since $i^*(G) \geq n/(\Delta+1)$, we have that for all $\lambda \geq 8^{(\Delta+1)/\delta^2}$
\begin{align*}
\mb{P}_{\lambda}[Y\leq (1-\delta/2)i^*(G)] &= \frac{1}{Z_G(\lambda)}\sum_{k=0}^{(1-\delta/2)i^*(G)}\binom{n}{k}\lambda^{k} \leq \frac{1}{Z_G(\lambda)}2^{n} \lambda^{(1-\delta/2)i^*(G)} \\
&\leq \frac{1}{Z_G(\lambda)}\lambda^{i^*(G)}\cdot 2^{n}\lambda^{-\delta n/(2\Delta + 2)} \leq \delta/4
\end{align*}
from which we see that $n\alpha_G(8^{(\Delta+1)/\delta^2}) \geq (1-\delta)i^*(G)$; and 
(iii) for all $\lambda > 0$ 
\[\frac{d}{d\lambda}\alpha_G(\lambda) = \frac{1}{n\lambda}\on{Var}_{\lambda}(Y) \leq C_{\delta, \Delta},\]
where the inequality follows from \cref{lem:var-bound}. 

As for the algorithmic claim, note that for each $t$ in the range, we can deterministically approximate $n\alpha_G(t/(C_{\ref{lem:find-fugacity}}n))$ to within an additive error of $1/4$ in time $n^{O_{\delta, \Delta}(1)}$ by \cref{thm:linear-time-cumulants}, so that we may find a $t$ with the desired property in the stated time. 
\end{proof}

\subsection{EPTAS from the LCLT}
\label{sub:EPTAS}

Recall that an EPTAS (efficient polynomial-time approximation scheme) is a PTAS (polynomial-time approximation scheme) with a running time of the form $f(\epsilon)n^{O(1)}$ i.e.~the degree of the polynomial in $n$ is independent of the error parameter $\epsilon$. By combining our LCLT with \cref{lem:find-fugacity} and \cref{thm:patel-regts,thm:linear-time-cumulants}, we immediately obtain an EPTAS for $i_k(G)$, for all $k$ bounded away from the relevant barrier. 

\begin{theorem}
\label{thm:eptas}
Let $\Delta \geq 3$ and $\delta \in (0,1/2)$. There exists a deterministic algorithm which, on input a graph $G = (V,E)$ on $n$ vertices of maximum degree at most $\Delta$, an integer $1\leq k \leq (1-\delta)\alpha_c(\Delta)n$, and an error parameter $\epsilon \in (0,1)$ outputs an $\epsilon$-relative approximation to $i_k(G)$ in time $n^{O_{\delta, \Delta}(1)}\exp(\tilde{O}_{\delta, \Delta}(\epsilon^{-1}/\sqrt{n}))$.

Moreover, if $G$ is claw-free, then the same conclusion holds for any $1\leq  k\leq (1-\delta)i^*(G)$. 
\end{theorem}
\begin{proof}
If $k < e^{-5}n/(\Delta+1)$, then an FPTAS for $i_k(G)$ is already implicit in \cite{davies2020proof}. Therefore, we may restrict our attention to $k \geq e^{-5}n/(\Delta + 1)$. We may also assume that $\epsilon^{-1} \leq c_{\delta, \Delta}\sqrt{n}/(\log{n})^3$ for a sufficiently small constant $c_{\delta, \Delta} > 0$; otherwise $\exp(\tilde{O}_{\delta, \Delta}(\epsilon^{-1}/\sqrt{n}))\geq 4^{n}$ so that exhaustive enumeration runs in the claimed time.  

Fix $k \geq e^{-5}n/(\Delta+1)$ which also lies in the specified range, let $t_k$ denote the corresponding value of $t$ given by \cref{lem:find-fugacity}, and let $\lambda_k = t/C_{\ref{lem:find-fugacity}}$ denote the corresponding fugacity. The upper bound on $\alpha_G'(\lambda)$ in the proof of \cref{lem:find-fugacity} shows that $\lambda_k = \Omega_{\delta, \Delta}(1)$. In particular, $\mu_{k} := \mu_{\lambda_k}, \sigma_{k} := \sigma_{\lambda_k}$ satisfy $\mu_k, \sigma^2_k = \Theta_{\delta, \Delta}(n)$. Moreover, since $|k - \mu_k| \leq 1/2$,  it follows by \cref{thm:independent-lclt2} that 
\begin{align*}
    \frac{i_k(G) \lambda_k^k}{Z_G(\lambda_k)} 
    &= \mb{P}[Y = k]\\
    &= \sigma_k^{-1}\mc{N}((k-\mu_k)/\sigma_k) \pm O_{\Delta, \delta}((\log n)^{5/2}/n)\\
    &= (1\pm \epsilon/1000)\cdot \sigma_k^{-1}\mc{N}((k-\mu_k)/\sigma_k)\\
    &= (1\pm \epsilon/500)\cdot (\sqrt{2\pi} \sigma_k)^{-1}. 
\end{align*}

Therefore, letting $\wh{Z}_G(\lambda_k)$ and $\wh{\sigma}_k$ denote  $\epsilon/1000$-relative approximations to $Z_G(\lambda_k)$ and $\sigma_{k}$, both of which can be computed deterministically in time $(n/\epsilon)^{O_{\delta, \Delta}(1)}$ by \cref{thm:patel-regts,thm:linear-time-cumulants}, it follows that
\[\wh{i_k}(G) = \lambda_k^{-k} \cdot \wh{Z}_G(\lambda_k)\cdot  (\sqrt{2\pi}\wh{\sigma}_k)^{-1} \]
is an $\epsilon$-relative approximation to $i_k(G)$, as desired. 
\end{proof}

\subsection{FPTAS from the proof of the LCLT}

In the previous subsection, we saw how our LCLT gives an EPTAS for $i_k(G)$, for all $k$ bounded away from the relevant barrier. Extending this to an FPTAS requires bypassing the error term $\tilde{O}(1/\sigma^2)$ in \cref{thm:independent-lclt2}. For this, instead of approximating $\mb{P}[Y=k]$ by $\sigma^{-1}\mc{N}((k-\mu)/\sigma)$, we will directly approximate $\mb{P}[Y=k]$ to the desired accuracy. The key ingredient required for this is the following.    
\begin{lemma}
\label{lem:approximate-characteristic-function}
Fix $\Delta \geq 3$, $\delta \in (0,1/2)$, and a parameter $C \geq 1$. There exists a deterministic algorithm which, on input a graph $G = (V,E)$ on $n$ vertices of maximum degree at most $\Delta$, $n^{-2} \leq \lambda \leq (1-\delta)\lambda_c(\Delta)$, an error parameter $\epsilon \in (0,1/\sqrt{n})$, and $t \in [-C\sqrt{\log{1/\epsilon}}, C\sqrt{\log{1/\epsilon}}]$ outputs an $\epsilon^{50}$-relative, $\epsilon^{50}$-additive approximation to $\mb{E}_{\lambda}[e^{itX}]$ in time $(n/\epsilon)^{O_{\delta,\Delta, C}(1)}$, where $X = (Y-\mu_{\lambda})/\sigma_{\lambda}$.

Moreover, if $G$ is claw-free, the same conclusion holds for all $\lambda \geq n^{-2}$ with running time $(n/\epsilon)^{O_{\Delta, C, |\lambda|}(1)}$.
\end{lemma}

\begin{proof}
We provide the proof for general graphs; the proof for claw-free graphs follows by straightforward modifications. For convenience of notation, we denote $\mu_{\lambda}$ and $\sigma_{\lambda}$ by $\mu$ and $\sigma$. By \cref{thm:patel-regts,thm:linear-time-cumulants}, we can compute $\epsilon^{100}$-relative approximations $\wh{Z}$, $\wh{\mu}$, and $\wh{\sigma}$ to $Z_G(\lambda)$, $\mu$, and $\sigma$ in time $(n/\epsilon)^{O_{\delta, \Delta}(1)}$.
Note that
\begin{align*}
    \mb{E}_{\lambda}[e^{itX}] 
    &= e^{-it\mu/\sigma}\cdot \mb{E}_{\lambda}[e^{itY/\sigma}]\\
    &= e^{-it\mu/\sigma}\cdot \frac{Z_G(\lambda e^{it/\sigma})}{Z_G(\lambda)}.
\end{align*}
We have the following two cases.

\textbf{Case I: }$\lambda \leq \frac{1}{10\Delta}$. In this case, for any $t \in \mb{R}$, $\lambda e^{it/\wh{\sigma}} \in \mc{R}_{\delta, \Delta}$, where $\mc{R}_{\delta, \Delta}$ is the zero-free region in \cref{thm:zero-free}. 

\textbf{Case II: }$\lambda > \frac{1}{10\Delta}$. In this case, we may assume that $\lambda e^{it/\wh{\sigma}} \in \mc{R}_{\delta, \Delta}$ for all $t \in \mb{R}$, $|t|\leq C\sqrt{\log{1/\epsilon}}$. Otherwise, since $\wh{\sigma}^2 = \Omega_{\delta, \Delta}(n)$ by \cref{lem:var-bound}, it follows that $\epsilon^{-1} = \exp(\Omega_{\delta, \Delta, C}(n))$, so that exhaustive enumeration runs in the claimed time.

Hence, in either case, it follows from \cref{thm:patel-regts} that an $\epsilon^{100}$-relative approximation to $Z_G(\lambda e^{it/\wh{\sigma}})$, which we denote by $\wh{Z}_{t}$, can be computed in time $(n/\epsilon)^{O_{\delta, \Delta}(1)}$.  

We claim that the output
\[e^{-it\wh{\mu}/\wh{\sigma}}\cdot \frac{\wh{Z}_t}{\wh{Z}}\]
is an approximation to $\mb{E}_{\lambda}[e^{itX}]$ of the desired quality. Indeed, by the assumed upper bounds on $|t|$ and $\epsilon$ and the assumed lower bound on $\lambda$, we have for all $n$ sufficiently large (depending on $\delta, \Delta, C$) that
\begin{align*}
    \frac{e^{it/\sigma}}{e^{it/\wh{\sigma}}} = e^{it(\wh{\sigma} - \sigma)/(\sigma \wh{\sigma})} = e^{i\theta}, \quad \quad \frac{e^{-it\mu}}{e^{-it\wh{\mu}}} = e^{i\theta'}
\end{align*}
for some $\theta, \theta' \in \mb{R}$ satisfying $|\theta|, |\theta'| \leq \epsilon^{97}$.
Therefore,  $\frac{e^{it\wh{\mu}/\wh{\sigma}}}{e^{it\mu/\sigma}} = e^{i\theta''}$ for some $\theta''\in \mb{R}$, $|\theta''| \leq \epsilon^{90}$ and
\begin{align*}
   |Z_G(\lambda e^{it/\wh{\sigma}}) - Z_G(\lambda e^{it/\sigma})| &= \left|\sum_{k=0}^{n}i_k(G)\lambda^{k}(e^{itk/\wh{\sigma}} - e^{itk/\sigma})\right|
   \leq \sum_{k=0}^{n}i_k(G)\lambda^{k}|e^{ik\theta} - 1|\\
   &\leq \epsilon^{94}\sum_{k=0}^{n}i_k(G)\lambda^{k} = \epsilon^{94}Z_G(\lambda).
\end{align*}
Hence,
\begin{align*}
    \mb{E}_{\lambda}[e^{itX}] 
    &= 
    \left(e^{i\theta''}e^{-it\wh{\mu}/\wh{\sigma}}\right)\cdot \frac{Z_G(\lambda e^{it/\wh{\sigma}})}{Z_G(\lambda)} + e^{-it\mu/\sigma}\cdot \frac{Z_G(\lambda e^{it/\sigma}) - Z_G(\lambda e^{it/\wh{\sigma}})}{Z_G(\lambda)}\\
    &=  \left(e^{i\theta''}e^{-it\wh{\mu}/\wh{\sigma}}\right)\cdot e^{\pm 2\epsilon^{100}}e^{\pm i \epsilon^{100}}\frac{\wh{Z}_t}{\wh{Z}} + w_{t}
\end{align*}
where $w_t \in \mb{C}$ with $|w| \leq \epsilon^{94}$, as desired.  
\end{proof}

Given the preceding lemma, \cref{thm:independent-fptas,thm:matching-fptas} follow from the proof of \cref{thm:independent-lclt2}. 

\begin{proof}[Proof of \cref{thm:independent-fptas,thm:matching-fptas}]
Again, we provide the details only for general graphs i.e.~\cref{thm:independent-fptas}; the argument for claw-free graphs is essentially the same. Let $k \geq 1$ be as in the statement of the theorem. Moreover, we may assume that $\epsilon \in (0,1/n)$, since the statement for larger values of $\epsilon$ follows from the statement for $\epsilon = 1/n$. Let $t_k$ denote the integer returned by \cref{lem:find-fugacity} and let $\lambda:=\lambda_k = t_k/C_{\ref{lem:find-fugacity}}$ denote the corresponding fugacity. Since $\mu := \mu_{\lambda_k} \geq 1/2$, it follows as in the proof of \cref{lem:find-fugacity} that $\lambda = \Omega_{\delta, \Delta}(1/n)$.

Let $\sigma := \sigma_{\lambda_k}$, $X = (Y-\mu)/\sigma$, and $\gamma = \min(\pi \sigma, C\sqrt{\log{1/\epsilon}})$, where $C$ is a sufficiently large constant depending on $\Delta, \delta$.  Let $\wh{Z}, \wh{\mu}, \wh{\sigma}$ be $\epsilon^{100}$-relative approximations to $Z_G(\lambda), \mu, \sigma$, which can be found deterministically in time $(n/\epsilon)^{O_{\delta, \Delta}(1)}$  by \cref{thm:patel-regts,thm:linear-time-cumulants}.

Let $x = (k-\mu)/\sigma$ and $\wh{x} = (k-\wh{\mu})/\wh{\sigma}$. 
Then, 
as in the proof of \cref{thm:independent-lclt2}, we have 
\begin{align*}
     \mb{P}_{\lambda}[X=x] 
    &= \frac{1}{2\pi \sigma}\int_{-\pi \sigma}^{\pi \sigma}\mb{E}_{\lambda}[e^{itX}]e^{-itx}dt \\
    &= \pm \epsilon^{100} + \Re\left(\frac{1}{2\pi \sigma}\int_{-\gamma}^{\gamma}\mb{E}_{\lambda}[e^{itX}]e^{-itx}dt\right) \\
    &= \pm 2\epsilon^{75} + \Re\left(\frac{\epsilon^{100}}{2\pi \sigma} \sum_{\ell = -\gamma\epsilon^{-100}}^{\gamma \epsilon^{-100}}\mb{E}_{\lambda}[e^{i \epsilon^{100}\ell X}]e^{-i\epsilon^{100}\ell x}\right)\\
    &= \pm 3\epsilon^{75} + \Re\left(\frac{\epsilon^{100}}{2\pi \sigma} \sum_{\ell = -\gamma\epsilon^{-100}}^{\gamma \epsilon^{-100}}\mb{E}_{\lambda}[e^{i \epsilon^{100}\ell X}]e^{-i\epsilon^{100}\ell \wh{x}}\right),
\end{align*}
where the first line uses the Fourier inversion formula for lattices, the second line uses the definition of $\gamma$, \cref{lem:high-fourier}, and \cref{lem:var-bound}, the third line uses the upper bound on $\epsilon$ and the lower bound on $\lambda$, and the last line uses the lower bound on $\lambda$.  

Next, for each integer $\ell \in [-\gamma \epsilon^{-100}, \gamma \epsilon^{-100}]$, let $\wh{Z}_{\ell}$ denote an $\epsilon^{100}$-additive, $\epsilon^{100}$-relative approximation to $\mb{E}_{\lambda}[e^{i\epsilon^{100}\ell X}]$. By \cref{lem:approximate-characteristic-function}, these approximations may be found deterministically in time $(n/\epsilon)^{O_{\delta, \Delta}(1)}$. Then, 
\begin{align*}
    (1\pm \epsilon^{75})\frac{i_k(G)\lambda^k}{\wh{Z}} &= \frac{i_k(G)\lambda^k}{Z_G(\lambda)} = \mb{P}_{\lambda}[X=x]\\ 
    &= \pm 3\epsilon^{75} + \Re\left(\frac{\epsilon^{100}}{2\pi \sigma} \sum_{\ell = -\gamma\epsilon^{-100}}^{\gamma \epsilon^{-100}}\mb{E}_{\lambda}[e^{i \epsilon^{100}\ell X}]e^{-i\epsilon^{100}\ell \wh{x}}\right)\\
    &= \pm 4\epsilon^{75} + (1\pm \epsilon^{50})\Re\left(\frac{\epsilon^{100}}{2\pi \wh{\sigma}} \sum_{\ell = -\gamma\epsilon^{-100}}^{\gamma \epsilon^{-100}}\wh{Z}_{\ell}e^{-i\epsilon^{100}\ell \wh{x}}\right). 
\end{align*}

Finally, since $\mb{P}_{\lambda}[X=x] = \mb{P}_{\lambda}[Y = k] = \Omega_{\delta, \Delta}(1/\sigma)$, we have that
\begin{align*}
(1\pm \epsilon^{50})\frac{i_k(G)\lambda^k}{\wh{Z}}  =   (1\pm \epsilon^{50})\Re\left(\frac{\epsilon^{100}}{2\pi \wh{\sigma}} \sum_{\ell = -\gamma\epsilon^{-100}}^{\gamma \epsilon^{-100}}\wh{Z}_{\ell}e^{-i\epsilon^{100}\ell \wh{x}}\right),
\end{align*}
so that the quantity
\[\lambda^{-k}\cdot \wh{Z}\cdot\Re\left(\frac{\epsilon^{100}}{2\pi \wh{\sigma}} \sum_{\ell = -\gamma\epsilon^{-100}}^{\gamma \epsilon^{-100}}\wh{Z}_{\ell}e^{-i\epsilon^{100}\ell \wh{x}}\right),\]
which can be computed deterministically in time $(n/\epsilon)^{O_{\delta, \Delta}(1)}$, is an $\epsilon$-relative approximation to $i_k(G)$. \qedhere

\end{proof}

\section{Randomized algorithms}
\label{sec:RandAlgorithms}

\subsection{A quasi-linear time sampling algorithm}
We proceed to the proof of \cref{thm:faster-sampling}. We will prove the result for independent sets, noting that the proof for matchings follows identically.

We will need the following preliminary lemma, which is an efficient randomized version of \cref{lem:find-fugacity}. 
\begin{lemma}\label{lem:lambda-determine}
Let $\Delta \geq 3$ and $\delta \in (0,1/2)$. There exists a constant $c_{\ref{lem:lambda-determine}}(\delta, \Delta) > 0$ and a randomized algorithm with the following property: for any graph $G = (V,E)$ on $n$ vertices of maximum degree at most $\Delta$, for any $1\leq k \leq (1-\delta)n\alpha_c(\Delta)$, and for any $\epsilon \in (0,1)$, the algorithm outputs $\lambda \in (0, (1-c_{\ref{lem:lambda-determine}})\lambda_c(\Delta)]$ satisfying 
\[|\mb{E}_{\lambda}{Y}-k|\le \sqrt{\on{Var}_{\lambda}{Y}}\]
with probability $1-\epsilon$. The running time of the algorithm is $O_{\Delta,\delta}(n\log(n/\epsilon)(\log n)^3)$.
\end{lemma}
\begin{remark}
If $k \geq n/(3\Delta)$, then it is easily seen that $\lambda \geq 1/(100\Delta)$. 
\end{remark}
\begin{proof}
If $1\leq k \leq \sqrt{n}$, then it follows from \cref{corEVbounds} that the deterministic choice $\lambda = k/n$ satisfies the desired conclusion. Therefore, it suffices to consider the case $\sqrt{n} \leq k \leq (1-\delta)n\alpha_c(\Delta)$.  
Let $\Lambda = [n^{-1/3}, (1-c_{\ref{lem:lambda-determine}})\lambda_c(\Delta)] \cap (\mb{Z}/n^{2})$. By \cref{lem:find-fugacity}, there exists $\lambda \in \Lambda$ satisfying 
\[|\mb{E}_{\lambda}Y - k| \leq 1/2. \]
Further, for any $\lambda \in \Lambda$, it follows from \cref{thm:clt-gen} and \cref{lem:var-bound} that
\[|\mb{E}_{\lambda}Y - \on{med}_{\lambda}Y| = \tilde{O}_{\delta, \Delta}(1),\]
where $\on{med}$ denotes the median. Since $\on{Var}_{\lambda}Y = \Theta_{\delta, \Delta}(n\lambda)$ by \cref{lem:var-bound}, it follows that there exists $\lambda \in \Lambda$ satisfying
\[|\on{med}_{\lambda}Y - k| \leq \frac{1}{2}\sqrt{\on{Var}_{\lambda}Y}\]
and it suffices to output such a $\lambda \in \Lambda$.

For any $\lambda \in \Lambda$, there is a randomized algorithm to estimate $\on{med}_{\lambda}Y$ to within $\sqrt{\on{Var}_{\lambda}Y}/2$ additive error which succeeds with probability $1- (\epsilon/n^3)$ and runs in time $O(n\log(n/\epsilon)(\log n)^2)$. Indeed, by \cref{thm:glauber,thm:clt-gen}, \cref{lem:var-bound} and the Chernoff bound, this may be accomplished by taking the median of $O_{\delta, \Delta}(\log(n/\epsilon))$ independent runs of the Glauber dynamics, each for $\Theta_{\Delta, \delta}(n\log n)$ steps. Therefore, running binary search with the above primitive takes time $O(n\log(n/\epsilon)(\log n)^3)$ and except with probability $1-\epsilon$, outputs $\lambda \in \Lambda$ satisfying the desired conclusion.
\end{proof}

We now present our sampling algorithm.  For simplicity of notation, we will restrict attention to the case $k > n/(3\Delta)$; the case $k \leq n/(3\Delta + 1)$ is already handled by the down-up walk in \cite{bubley1997path} with asymptotically optimal running time. Below, $c_{\Delta, \delta}, C_{\Delta, \delta}, c'_{\Delta, \delta}, C'_{\Delta, \delta}$ are constants depending on $\Delta, \delta$ whose values can be determined using \emph{a priori} analysis. We will assume that $\epsilon \geq \exp(-n/C'_{\Delta, \delta})$; for smaller $\epsilon$, it follows from \cref{thm:glauber,thm:independent-lclt2} that rejection sampling for the hard-core model at the fugacity $\lambda$ given by \cref{lem:lambda-determine} outputs a distribution within $\epsilon$-TV distance of the uniform distribution on $\mc{I}_k[G]$ in time
\begin{align*}
    O_{\delta, \Delta}(n\log(n/\epsilon)\cdot \log(n)\cdot \sqrt{n}\log(1/\epsilon)) = O_{\delta, \Delta}(n\log(1/\epsilon)^{5/2}\log n). 
\end{align*}
For $1/100 \geq \epsilon \geq \exp(-n/C'_{\Delta, \delta})$, we use the following algorithm. 
\begin{itemize}
    \item (Preprocessing Step 1) Using \cref{lem:lambda-determine}, find $\lambda \in \Lambda$ such that \[|\mb{E}_{\lambda}{Y}-k|\le \sqrt{\on{Var}_{\lambda}{Y}}.\]
    \item (Preprocessing Step 2) Using \cref{lem:well-separated}, find $S\subseteq V$ such that $|S| = \Omega(n/\Delta^3)$ and such that the vertices in $S$ are distance at least $4$ apart. Let $T = \{v:v \text{ such that} \on{dist}(v,S)\ge 2\}$.
    \item (Preprocessing Step 3) Find an independent set $I_0$ of size $k$ using a similar greedy algorithm as in \cref{lem:well-separated}. 
    \item (Initialize Core) Run Glauber dynamics for the hard-core model on $G$ at fugacity $\lambda$ for $\Theta_{\Delta, \delta}(n\log(n/\epsilon))$-steps to obtain an independent set $I'$. Let $J = I'\cap T$. 
    \item (Set parameters) Fix $p = c_{\Delta,\delta} \sqrt{n/\log(1/\epsilon)}$, where $c_{\Delta, \delta} > 0$ is sufficiently small. For each $v \in S$ and for each $K \in \mc{I}(G[v\cup N(v)])$, use exhaustive enumeration to compute 
    \[p_{v,K} = \mb{P}_{\lambda}[I \cap {N(v)} = K \mid I \cap T = J].\]
    Let
    \[q_{v} = \min(p_{v,\emptyset}, p_{v, K_v}),\]
    where $K_{v} \in \mc{I}[N(v)]$ denotes the independent set with $v$ occupied and all other vertices unoccupied. 
    Let
    \[q = \min_{v \in S}q_{v}.\]
    For each $v \in S$, let $W_{v}$ be a random variable taking values in $\mc{I}[N(v)] \cup \{?\}$ which takes on the value $?$ with probability $2q$, $\emptyset$ with probability $p_{v,\emptyset} - q$, $K_{v}$ with probability $p_{v,K_{v}} - q$, and $K\notin\{\emptyset, K_{v}\}$ with probability $p_{v, K}$. 
    \item (Resample Neighborhoods Step 1) For each $v \in S$, independently, sample $W_{v}$. If $W_{v} \neq ?$, then set $I \cap N(v) = W_{v}$. Let $S^{\ast} = \{v : W_{v} = ?\}$. If $|S^{\ast}| \leq c'_{\Delta, \delta}n$, where $c'_{\Delta, \delta} > 0$ is sufficiently small, then proceed to the final step. Else, proceed to the next step. 
    \item (Resample Neighborhoods Step 2) Let $\ell$ be the current number of vertices chosen and let $k - \ell = k'$. With probability $p\binom{|S^{\ast}|}{k'}2^{-|S^{\ast}|}$ (here, we use the convention that the binomial coefficient is $0$ is $k'\notin [0, |S^{\ast}|]$) sample a random subset of $S^{\ast}$ of size $k'$ and set $I \cap N(v) = K_{v}$ for the vertices in the subset and $I \cap N(v) = \emptyset$ for the vertices not in the subset; with the remaining probability, proceed to the next step. 
    \item Repeat all steps after preprocessing at most at most $C_{\Delta,\delta}\log(1/\epsilon)^{3/2}$ times for sufficiently large $C_{\Delta, \delta} > 0$; if no valid sample after has been reached, output $I_0$.
\end{itemize}

We now show that the above algorithm satisfies the assertion of \cref{thm:faster-sampling} for $k$ and $\epsilon$ in the specified range. 
\begin{proof}[Proof of \cref{thm:faster-sampling}]

Before analyzing the correctness of the algorithm, let us quickly bound its running time. By \cref{lem:find-fugacity,lem:well-separated}, the preprocessing steps take time $O_{\Delta, \delta}(n\log(n/\epsilon)(\log n)^3)$. Note that the Preprocessing Step 3 can indeed be accomplished using the greedy algorithm in \cref{lem:well-separated} since $\alpha_c(\Delta) \leq 1/(\Delta+1)$. Each run of Initalize Core takes time $O_{\Delta, \delta}(n\log(n/\epsilon)\log n)$ since each step of Glauber dynamics takes time $O_{\Delta, \delta}(\log n)$ to implement. The step Set parameters takes time $O_{\Delta, \delta}(n)$. Resample Neighborhoods Step 1 takes time $O_{\Delta, \delta}(n)$. In Resample Neighborhoods Step 2, we can compute the probability in time $O(n(\log n)^2)$  by \cite{FT15} and then sample from the hypergeometric distribution using sampling without replacement in time $O(n\log{n})$. Thus, we see that the running time of the algorithm is
\[O_{\delta, \Delta}(n\log(n/\epsilon)(\log n)^3 + n\log(n/\epsilon)\log n\log(1/\epsilon)^{3/2}).\]

We now proceed to the proof of correctness. The idea is that the algorithm may be viewed as implementing rejection sampling where the base sampler outputs $I$ according to the distribution $\mu_{G,\lambda}(\cdot \mid |I| \equiv k \bmod p)$ in time $\tilde{O}_{\delta, \Delta}(n)$. This leads to a $\tilde{O}_{\delta, \Delta}(n)$ time algorithm for approximately sampling from the uniform distribution on $\mc{I}_k[G]$ since by \cref{lem:lambda-determine} and \cref{thm:independent-lclt2}, $\mb{P}_{\lambda}[|I| = k \mid |I| \equiv k \bmod p] = \tilde{\Omega}_{\delta, \Delta}(1)$. 

To formalize this, we begin by noting that for all $t \in \{0,\dots, p-1\}$ and for all $J \in \mc{I}[G[T]]$, it follows from the calculations and notation in the proof of \cref{lem:high-fourier} that
\begin{align*}
    \left|\mb{P}_{\lambda}[|I| \equiv t \bmod p \mid I \cap T = J] - \frac{1}{p}\right|
    &= \left|\mb{P}_{\lambda}[X_{1} + \dots + X_{s} \equiv t - |J| \bmod p] - \frac{1}{p}\right|\\
    &\leq \frac{1}{p}\sum_{\ell = 1}^{p-1}\exp(-\Omega_{\delta, \Delta}(n\ell^2/p^2))\\
    & \leq \frac{\epsilon}{100p},
\end{align*}
provided that $c_{\Delta, \delta} > 0$ is chosen to be sufficiently small. Hence, for any $J$
\begin{align*}
    \mb{P}_{\lambda}[I \cap T = J \mid |I| \equiv k \bmod p] = \mb{P}_{\lambda}[I \cap T = J](1\pm\epsilon/50), 
\end{align*}
so that up to an $\epsilon/49$-TV distance, the distribution of the set $J$ in Initialize Core is $\mu_{G,\lambda}[I \cap T = \cdot \mid |I| \equiv k \bmod p]$.

Next, for any realisation $\vec{w} = (w_v)_{v\in S}$ of $\vec{W} = (W_v)_{v\in S}$, let $S^*(\vec{w})$ denote the corresponding subset and let $\mu_{\vec{w}}$ denote the uniform distribution on $\{0,1\}^{S^*(\vec{w})}$ with $0$ denoting unoccupied and $1$ denoting occupied. Then, sampling from the conditional distribution $\mu_{G,\lambda}(\cdot \mid I \cap T = J)$ is equivalent to first sampling $\vec{W} = \vec{w}$ according to the distribution specified in Set Parameters and then resampling the $?$ according to $\mu_{\vec{w}}$. For any realisation $\vec{w}$ of $\vec{W}$ with $|S^*(\vec{w})| \geq c'_{\Delta, \delta}n$, it follows as above that for any $t\in \{0,\dots,p-1\}$, 
\begin{align*}
    \left|\mb{P}_{\lambda}[|I|\equiv t \bmod p \mid I \cap T = J, \vec{W} = \vec{w}] - \frac{1}{p}\right|
    &\leq \max_{t'} \left|\mb{P}[\on{Binomial}(|S^*(\vec{w})|, 1/2) =t' \bmod p]-\frac{1}{p}\right|\\
    &\leq \frac{\epsilon}{100p},
\end{align*}
provided that $c_{\Delta, \delta} > 0$ is chosen to be sufficiently small compared to $c'_{\Delta, \delta}$. Since $q = \Omega_{\delta, \Delta}(1)$, it follows by the Chernoff bound that
\[\mb{P}_{\lambda}[|S^*(\vec{w})| \leq c'_{\Delta, \delta}n \mid I \cap T = J] \leq \exp(-\Omega_{\Delta, \delta}(n)) \leq \epsilon/n^{3}\]
for all sufficiently large $n$. 
This shows that 
\begin{enumerate}[(P1)]
\item The probability of moving directly to the final step from Resample Neighborhoods Step 1 is at most $\epsilon/n^3$, and 
\item Up to an $\epsilon/49$-TV distance, the $(W_v)_{v\in S}$ in Resample Neighborhoods Step 1 follow the distribution $(W_v)_{v\in S} \mid \{|I|\equiv k \bmod p, I \cap T = J\}$.
\end{enumerate} 

Now, we analyze Resample Neighborhoods Step 2.
Observe that for any realisation $\vec{W} = \vec{w}$ with $|S^{\ast}(\vec{w})| \geq c'_{\Delta, \delta}n$,
\begin{align*}
    \mb{P}_{\lambda}[|I| = k \mid I \cap T = J, \vec{W} = \vec{w}, |I| \equiv k \bmod p]
    &= \frac{\mb{P}_{\lambda}||I| = k \mid I \cap T = J, \vec{W} = \vec{w}]}{\mb{P}_{\lambda}[|I| \equiv k \bmod p \mid I \cap T = J, \vec{W} = \vec{w}]}\\
    &= \frac{\mb{P}[\on{Binomial}(|S^{\ast}(\vec{w})|,1/2) = k']}{(1\pm \epsilon/100)/p}\\
    &= (1\pm \epsilon/50)p \binom{|S^*(\vec{w})|}{k'}2^{-|S^*(\vec{w})|}.
\end{align*}
Therefore, Resample Neighborhoods Step 2 samples from a distribution within $\epsilon/49$-TV distance to the conditional distribution $\mu_{\lambda, G}[\cdot \mid I \cap T = J, \vec{W} = \vec{w}, |I|\equiv k \mod p]$ and rejects the sample if $|I| \neq k$. 

So far, we have not used the property of $\lambda$ guaranteed by \cref{lem:lambda-determine}. This will now be used to show that the probability of the steps Initialize Core through Resample Neighborhoods Step 2 outputting an independent set of size $k$ is $\Omega_{\Delta, \delta}(1/\sqrt{\log(1/\epsilon)})$. To see this, note that by taking the expectation over $J$ and $\vec{W}$ on both sides in the display equation above and using (P1) and \cref{thm:independent-lclt2}, we have that 
\begin{align*}
\mb{E}_{J, \vec{W}}\left[\binom{S^*(\vec{W})}{k'}p2^{-|S^*(\vec{W})|} \mid |S^*(\vec{W})| \geq c'_{\Delta, \delta}n\right]
&= (1\pm \epsilon/25)\mb{P}_{\lambda}[|I| = k \mid |I| \equiv k \bmod p] \\
&= \Omega_{\delta, \Delta}(p/\sqrt{n})\\
&= \Omega_{\delta, \Delta}(1/\sqrt{\log(1/\epsilon)}).
\end{align*}
If $|S^*(\vec{W})| \geq c'_{\Delta, \delta}n$, then the quantity inside the expectation is bounded by $O_{\delta, \Delta}(1/\sqrt{\log(1/\epsilon)})$. Hence, by the reverse Markov inequality, with probability $\Omega_{\delta, \Delta}(1)$ (over the choice of $J, \vec{W}$), the quantity inside the expectation is $\Omega_{\delta, \Delta}(1/\sqrt{\log(1/\epsilon)})$.  

To summarize, we have shown the following: a single run of Initialize Core through Resample Neighborhoods Step 2 produces an output with probability $\Omega_{\delta, \Delta}(1/\sqrt{\log(1/\epsilon)})$ and the distribution of this output is within $\epsilon/5$ in TV-distance from the uniform distribution on $\mc{I}_k(G)$. Therefore, by the Chernoff bound, repeating this procedure independently $C_{\Delta, \delta}\log(1/\epsilon)^{3/2}$ times for $C_{\Delta, \delta}$ produces an output from a distribution on $\mc{I}_k(G)$ which is within $\epsilon$ in TV-distance of the uniform distribution on $\mc{I}_k(G)$.  \qedhere
\end{proof}

\subsection{A faster FPRAS} The FPRAS for $i_k(G)$ and $m_k(G)$ is substantially simpler than the sampling algorithm above. As before, we will present the proof only for $i_k(G)$ with the proof for $m_k(G)$ being similar. 

\begin{proof}[Proof of \cref{thm:faster-fpras}]
We have the following two cases:

\textbf{Case I: $1 \leq k \leq c_{\Delta}\sqrt{n}$}, where $c_{\Delta} > 0$ is a sufficiently small constant which can be determined \emph{a priori}. Let
\[p_k := \mb{P}[J \in \mc{I}_k],\]
where $J$ is a uniformly random subset of $V$  of size exactly $k$. By the union bound, it follows that
\begin{align*}
    p_k &= 1 - O\left(n\Delta\cdot \frac{k^2}{n^2}\right) \geq \frac{1}{2},
\end{align*}
provided that $c_{\Delta}$ is sufficiently small. 
Since 
\[i_k(G) = \binom{n}{k}p_k,\]
it suffices to obtain an $\epsilon$-relative approximation of $p_k$. Let $S_{1},\dots, S_{\ell}$ denote independent samples from the uniform distribution on size $k$ subsets of $V$. Then, by the Chernoff bound, 
\[\frac{\mbm{1}[S_1 \in \mc{I}_k(G)] + \dots + \mbm{1}[S_\ell \in \mc{I}_k(G)]}{\ell} = (1\pm \epsilon)p_k\]
with probability at least $3/4$, provided that $\ell > C/\epsilon^2$ for a sufficiently large constant $C$. 

For the running time, note that sampling a uniformly random subset of size $k$ takes time $O(k\log{n})$, checking whether it is an independent set takes time $O_{\Delta}(k)$, and computing the binomial coefficient $\binom{n}{k}$ takes time $O(k\log{n}\log\log{n})$, so that the total running time is
\[O(k\log{n}\log\log{n}) + O_{\Delta}(k\epsilon^{-2}\log{n}).\]

\textbf{Case II: $c_{\Delta, \delta}\sqrt{n} \leq k \leq (1-\delta)\alpha_c(\Delta)n$}. In this case, we first use \cref{lem:find-fugacity} to find a suitable $\lambda$. By \cref{thm:independent-lclt2}, 
\begin{align*}
    p_k := \frac{i_k(G)\lambda^{k}}{Z_G(\lambda)} = \mb{P}_{\lambda}[|I| = k] = \Omega_{\Delta, \delta}\left(\frac{1}{\sqrt{n\lambda}}\right).
\end{align*}
Let $I_1,\dots, I_{\ell}$ denote independent samples obtained by running the Glauber dynamics for the hard-core model at fugacity $\lambda$ for $O_{\Delta, \delta}(n\log n)$ steps. Then, by the Chernoff bound,
\begin{align*}
    \frac{\mbm{1}[I_1 \in \mc{I}_k(G)] + \dots +  \mbm{1}[I_{\ell} \in \mc{I}_k(G)]}{\ell} = (1\pm \epsilon/4)p_k
\end{align*}
with probability at least $3/4$, provided that $\ell > C_{\Delta, \delta}\epsilon^{-2}\sqrt{n\lambda}$, for a sufficiently large constant $C_{\Delta, \delta}$. 

For the running time, sampling each $I_i$ takes time $O_{\Delta, \delta}(n(\log n)\log(n/\epsilon))$, finding its size takes time $O_{\Delta}(n)$, computing $\lambda^{k}$ takes time $O(k\log{n}\log\log{n})$ and approximating $Z_G(\lambda)$ to within an $\epsilon/2$-relative approximation takes time $T$, which gives the desired conclusion. 
\end{proof}

\section{Cluster expansion}
\label{secCluster}

In this section, we treat the case of small activities $\lam$ using the cluster expansion, a classical tool from statistical physics.   The cluster expansion (or Mayer series~\cite{mayer1941molecular}) is a formal infinite series for $\log Z_G(\lam)$.  
For an introduction to the cluster expansion, see~\cite[Chapter 5]{friedli2017statistical}.  

We introduce the cluster expansion in the special case of the hard-core model on a graph $G$.  A \textit{cluster} $\Gamma$ is an ordered tuple of vertices from $G$.  The size of $\Gamma$, denoted $|\Gamma|$, is the length of the tuple.  The incompatibility graph of $\Gamma = (v_1, \dots, v_k)$, $H(\Gamma)$, is the graph with vertex set $\{v_1, \dots, v_k\}$ and an edge between $v_i, v_j$, $i \ne j$, if $v_i \in N(v_j) \cup \{v_j\}$ in $G$.   The Ursell function of a graph $H$ is the function
\[ \phi(H) = \frac{1}{|V(H)|!} \sum_{A \subseteq E(H): (V(H), A) \text{ connected}} (-1)^{|A|} \,. \]
The cluster expansion is the formal infinite power series
\[ \log Z_G(\lam) = \sum_{\Gamma} \phi(H(\Gamma)) \lam^{|\Gamma|} \, , \]
where the sum is over all clusters of vertices from $G$.   In fact, in this setting the cluster expansion is simply the Taylor series for $\log Z_G(\lambda)$ around $0$, with terms organized by clusters instead of grouping all terms of order $k$ together.

 To use the cluster expansion as an enumeration tool, it is essential to bound its rate of convergence.  We will use the convergence criteria of Koteck\'{y} and Preiss~\cite{kotecky1986cluster} (though the zero-freeness result of Shearer~\cite{shearer1985problem} along with the lemma of  Barvinok~\cite{barvinok2016combinatorics} on truncating Taylor series would also work).  
 This lemma bounds the additive error of truncating the cluster expansion after a given number of terms.
\begin{lemma}
\label{lemKPhardcore}
Let $G$ be a graph of maximum degree at most $\Delta$ on $n$ vertices and suppose $ 0 < \lam < \frac{e}{\Delta+1}$.  Then
\begin{equation}
\label{eqKPhc}
\sum_{\Gamma:  | \Gamma | \ge k}  \left| \phi(H(\Gamma)) \lam^{|\Gamma|} \right| \le n \left(  \lam e(\Delta+1)  \right) ^k  \,.
\end{equation}
\end{lemma}
\begin{proof}
This is a consequence of the main result of~\cite{kotecky1986cluster}.  We can express the hard-core model as a polymer model in the setting of~\cite{kotecky1986cluster} by defining  each vertex to be a polymer with weight $\lam$.  Taking $a(v) =1$ for all $v \in V(G)$ and $\exp(d(v)) = \frac{1}{\lam e (\Delta+1)}$, we have for all $v \in V(G)$,
\[  \sum_{u \in N(v) \cup \{v\}}  \lam e^{a(v) + d(v) } \le (\Delta+1) \lam e   \frac{1}{\lam e (\Delta+1) }=1 \,. \]
Then by the main theorem in~\cite{kotecky1986cluster}, for all $v \in V(G)$, 
\[ \sum_{\Gamma \ni v}  \left |\phi(H(\Gamma)) \lam^{|\Gamma|} \left(  \frac{1}{\lam e (\Delta+1) }\right)^{|\Gamma|} \right | \le 1\,. \]
Restricting the sum to clusters of size at least $k$ and summing over all $v \in V(G)$ gives~\eqref{eqKPhc}.
\end{proof}

The cluster expansion is a very convenient tool for studying the cumulants of the random variable $Y = | I|$ (see e.g.~\cite{dobrushin1996estimates,cannon2019bipartite,jenssen2020independent}).   In particular, when the cluster expansion converges, we have the formula
\begin{equation}
\label{eqCumulantCluster}
\kappa_k(Y) = \sum_{\Gamma} |\Gamma|^k   \phi(H(\Gamma)) \lam^{|\Gamma|} \,. 
\end{equation}

\begin{lemma}
\label{lemClusterKbound}
Fix $\Delta\ge 2, \delta >0$ and suppose $\lam \le \frac{1-\delta}{e (\Delta+1)}$. Then for all fixed $k \ge 1$, and all graphs $G$ of maximum degree $\Delta$ on $n$ vertices,
\[ \sum_{\Gamma} |\Gamma|^k \lam^k \phi(H(\Gamma)) =  n \lam + O_{k,\delta} (n \lam^2 \Delta^2) \,. \]
\end{lemma}
\begin{proof}
Since the contribution to the left-hand side from clusters of size $1$ is $n \lam$, it suffices to show that 
\[ \sum_{\Gamma:  | \Gamma | \ge 2}  |\Gamma|^k \lam^{|\Gamma|} | \phi(H(\Gamma))| =  O_{k,\delta} (n \lam^2 \Delta^2)  \,.\]
Applying \cref{lemKPhardcore} ,we have 
\begin{align*}
     \sum_{\Gamma:  | \Gamma | \ge 2}  |\Gamma|^k \lam^{|\Gamma|} | \phi(H(\Gamma))| &\le n \sum_{t \ge2} t^k (\lam e (\Delta+1))^t \le C_{k,\delta} n \lam^2  (\Delta+1)^2  \,. \qedhere
\end{align*}
\end{proof}

As in~\cite[Corollary 23]{jenssen2020independent}, we can immediately deduce bounds on the mean and variance and a central limit theorem from \cref{lemClusterKbound}.
\begin{corollary}
\label{corEVbounds}  Let $G$ be a graph of maximum degree $\Delta$.
If $ \lam \le \frac{1-\delta}{e (\Delta+1)}$, then the following hold: 
\begin{enumerate}
\item $\E_\lam Y =  n \lam + O_{\delta}(n \lam^2 \Delta^2)$.
\item $\var_{\lam} Y =  n \lam + O_{\delta}(n \lam^2 \Delta^2)$.
\end{enumerate}
If in addition $\lam n \to \infty$ as $n \to \infty$, then
\begin{enumerate}
\setcounter{enumi}{2}
\item The random variable $Y$ satisfies a central limit theorem.
\end{enumerate}
\end{corollary}
\begin{proof}
We prove these statements via the cumulants of  $Y$.  
\begin{align*}
\E_\lam  Y  = \kappa_1 ( Y)  & = \sum_{\Gamma} |\Gamma| \lam^{|\Gamma|} \phi(H(\Gamma)) =  n \lam +  O_{\delta}(n \lam^2 \Delta^2) \, .
\end{align*}
\begin{align*}
\var_\lam  Y = \kappa_2 ( Y)  & = \sum_{\Gamma} |\Gamma|^2 \lam^{|\Gamma|} \phi(H(\Gamma)) =  n \lam+ O_{\delta}(n \lam^2 \Delta^2) \, .
\end{align*}
Now let $X = (Y- \E_{\lam} Y)/\sqrt{\var_{\lam} Y}$.  By definition $\kappa_1(X) =0$ and $\kappa_2(X) =1$. All the cumulants of a standard Gaussian random variable are $0$ except for the second which is $1$,  and so to prove a central limit theorem it suffices to show that for any fixed $k \ge 3$, $\kappa_k (X) \to 0$ as $n \to \infty$.   We have 
\begin{align*}
|\kappa_k(X) |  &= \left | \sum_{\Gamma} \frac{|\Gamma|^k}{\var( Y)^{k/2} } \lam^{|\Gamma|} \phi(H(\Gamma))  \right | \\
&\le  \frac{ 1}{ \var( Y)^{3/2}} \sum_{\Gamma} |\Gamma|^k \lam^{|\Gamma|} |\phi(H(\Gamma)) | \\
&=(1+o(1)) (n \lam )^{-1/2} = o(1) \,. \qedhere
\end{align*}
\end{proof}

Note that if $\lam n \to \rho >0$ as $n \to \infty$, then \cref{lemClusterKbound}  shows that $\kappa_k(Y) \to \rho$ for each fixed $k$, which implies that $Y$ converges in distribution to a Poisson random variable of mean $\rho$. 

\begin{theorem}
\label{thmClusterLocal}
Fix $\Delta \ge 3$.  If $G_n$  is a sequence of graphs of maximum degree $\Delta $ on $n$ vertices and $n^{-1} \ll \lam  \le \frac{1}{100 \Delta^2}$, then the random variable $Y$ satisfies a local central limit theorem as $n \to \infty$. 
\end{theorem}
\begin{proof}
We follow the proof strategy of~\cite[Theorem 19]{jenssen2021independent} that proves a local central limit theorem in the setting of polymer models satisfying the equivalent of \cref{lemClusterKbound}. 

Let  $X = (Y- \E_\lam Y)/\sqrt{\var_{\lam} Y}$, and let $\phi_X(t) = \E_{\lam} e^{itX}$ be the characteristic function of $X$ (and $\phi_Y$ the characteristic function of $Y$).  

First we prove that there exists $c>0$ so that for all $t \in [-\pi, \pi]$, $|\phi_X(t)| \le e^{-ct^2 }$.  Using the cluster expansion we write
\begin{align*}
    \log \phi_Y(t) &= \sum_{\Gamma}  \left( e^{it|\Gamma| }-1 \right) \phi(H(\Gamma)) \lam^{|\Gamma|} \,,
\end{align*}
and so
\begin{align*}
    \mathrm{Re}  \log \phi_Y(t) &= \sum_{\Gamma} \left( \cos ( t |\Gamma|)-1  \right) \phi(H(\Gamma)) \lam^{|\Gamma|} \\
    &= n \lam (\cos(t) -1) + \sum_{\Gamma: |\Gamma|\ge 2} \left( \cos ( t |\Gamma|)-1  \right) \phi(H(\Gamma)) \lam^{|\Gamma|} \\
    &\le -\frac{t^2 n \lam }{5} + \sum_{\Gamma: |\Gamma|\ge 2} t^2 |\Gamma|^2  \phi(H(\Gamma)) \lam^{|\Gamma|}  \\
    &\le -\frac{t^2 n \lam }{5} + -\frac{t^2 n \lam }{10}    \le -\frac{t^2 n \lam }{10} \,,
\end{align*}
where we have used the bounds $\cos(t)-1 \le -t^2/5$ and $1-\cos(tx) \le (tx)^2$.
Exponentiating then gives  $|\phi_Y(t)| \le e^{-t^2 n \lam/10}$, which, along with \cref{corEVbounds}, implies that $| \phi_X(t) | \le e^{-c t^2}$ for some $c >0$.

To finish the proof we will apply \cref{lem:fourier-convert} with $\alpha = \E Y$ and $\beta = \sqrt{\var(Y)}$, which says
\begin{align*}
\sup_{x\in \mc{L}}|\beta\mc{N}(x) - \mb{P}[X=x]| &\le \beta\int_{-\pi/\beta}^{\pi/\beta}\big|\phi_X(t)-\phi_{\mc{Z}}(t)\big|dt + e^{-\pi^2/(2\beta^2)}  \\
&\le o(\beta) + \beta\int_{-\infty}^{\infty}\big|\phi_X(t)-\phi_{\mc{Z}}(t)\big|dt \,,
\end{align*}
and so to prove the LCLT  it suffices to show that 
\[ \int_{-\infty}^{\infty}\big|\phi_X(t)-\phi_{\mc{Z}}(t)\big|dt = o(1) \,. \]
By the central limit theorem of \cref{corEVbounds} we have that $\phi_X(t) \to \phi_{\mc{Z}}(t)$ as $n \to \infty$. Moreover $\big|   \phi_X(t)-\phi_{\mc{Z}}(t) \big |$ is an integrable function since it is bounded by $e^{-ct^2} + e^{-t^2/2}$ from the bound above.  Applying dominated convergence completes the proof.
\end{proof}

\section{Deterministic approximation of cumulants in linear time}
\label{secCumulants}

In this section we prove \cref{thm:linear-time-cumulants}. We prove the theorem first in the regime of cluster expansion convergence, then extend to more general zero-free regions. We prove the theorem for general graphs, noting that the proof for claw-free graphs is similar.

 \begin{lemma}
\label{thmAlgsmallLam}
For all graphs $G$ of maximum degree $\Delta$, all $0 < \lam \le   \frac{1- \delta}{e (\Delta+1)}$, and all fixed $k \ge 1$ there is a deterministic algorithm to give an $\epsilon \lam n$  additive approximation to $\kappa_k(Y)$.  The algorithm runs in time $O_{\Delta, \delta, k}(n \cdot (1/\epsilon)^{O_{\Delta, \delta}(1)})$. 
\end{lemma}
\begin{proof}  The algorithm will be to compute a truncation of the cluster expansion for $\kappa_k(Y)$. Recall that in the regime of cluster expansion convergence, we have 
\[\kappa_k(Y) = \sum_{\Gamma} |\Gamma|^k \phi(H(\Gamma)) \lam^{|\Gamma|}   \,.\]
Now let $T_t^{(k)} = \sum_{|\Gamma| < t} |\Gamma|^k \phi(H(\Gamma)) \lam^{|\Gamma|}$ be the truncation keeping only clusters of size less than $t$. By~\eqref{eqCumulantCluster} and~\eqref{eqKPhc} we have 
\begin{align*}
\left| \kappa_k(Y) - T_t^{(k)}  \right | &\le  n \sum_{ j \ge t} j^k \left(  \lam e(\Delta+1)  \right) ^j.
\end{align*}
By taking $t = \Omega \left( \frac{ \log (\epsilon) }{\log (\lam e (\Delta+1))} + \frac{k^2}{\delta^2}\right) $, we  have that $\left| \kappa_k (Y) - T_t^{(k)}  \right | \le n \lam \epsilon $.    This truncated cluster expansion  can be computed in time $n \cdot \exp ( O(t \log \Delta)) = O_{\Delta, \delta, k}(n \cdot (1/\epsilon)^{O_{\Delta, \delta}(1)})$ using the algorithm of~\cite{patel2017deterministic,helmuth2020algorithmic}.
\end{proof}

Next we give a general algorithm when  $\lam$ is  not necessarily in the regime of cluster expansion convergence.
\begin{proof}[Proof of \cref{thm:linear-time-cumulants}]
Since \cref{thmAlgsmallLam} covers the case of $\lam \leq 1/(2e(\Delta+1))$, we will assume here that $\lam \geq 1/(2e(\Delta+1))$. 
The general algorithm is an adaptation of the approximate counting algorithm of Barvinok and Patel and Regts via an expression for $\kappa_k(Y)$ in terms of derivatives of $\log Z_G(\lam)$.  The $k$th cumulant of $Y$ can be written as a linear combination of the first $k$ derivatives of $\log Z_G(\lam)$ in $\lam$, where (for $\lambda$ in the considered range) the size of the coefficients in the linear combination can be bounded in terms of only $k$ and $\Delta$. Hence, it suffices to give an $\epsilon n$ additive approximation to $\frac{d^k}{d \lam^k} \log Z_G(\lam)$ in the stated running time.

Let $\delta >0$ be small enough so that $\lam \in \mc{R}_{\delta, \Delta}$.
Following Barvinok~\cite{barvinok2016combinatorics} and Peters and Regts~\cite{PR19},  there is a polynomial $f$ of  degree $D= D(\delta,\Delta)$ that maps the unit circle in the complex plane into the region $\mc{R}_{\delta, \Delta}$, sending $0$ to $0$ and $1$ to $\lam$.  Let 
\[ \hat Z (y) = Z_G( f(y) ) \]
so that $\hat Z(1) = Z_G(\lam) $.  In particular, $\hat Z$ is a polynomial in $y$ of degree $N \le D n$.  Let $r_1, \dots, r_N$ denote the inverses of the roots of $\hat Z$ so that $\hat Z(y) = \prod_{i=1}^n (1-r_i y)$.   By \cref{thm:zero-free}, there is an $\eta = \eta(\delta,\Delta) \in (0,1)$ so that $|r_i| \le \eta$ for $i=1, \dots , N$. 

The first $k$ derivatives of $\log Z_G(\lam)$ with respect to $\lam$ can be written in terms of the first $k$ derivatives of $\log \hat Z(y)$ with respect to $y$ and those of $f$ with respect to $y$.  Using the chain rule we obtain
\[ \frac{d^k \log Z_G(\lam)}{d \lam^k}  =  \sum_{j=1}^k b_j \frac{ d^j \log \hat Z(y) }{ d y^j }   \]
where the coefficients $b_j$ depend only on the first $j$ derivatives of the bounded-degree polynomial $f$ and thus are bounded. For instance, we have 
\begin{align*}
    \frac{d \log Z_G(\lam)}{d \lam}  &=  \frac{ \frac{d \log \hat Z(y)}{dy} }{ \frac{d f(y)}{dy}  } 
    \intertext{and}
     \frac{d^2 \log Z_G(\lam)}{d \lam^2}  &=  \frac{ \frac{d^2 \log \hat Z(y)}{dy^2} }{ \frac{d f(y)}{dy}  }  - \frac{ \frac{d \log \hat Z(y)}{dy} \cdot  \frac{d^2 f(y)}{dy^2}  }{ \left( \frac{d f(y)}{dy}  \right)^2 } 
\end{align*}
In particular, it now suffices to compute an $\epsilon n $ additive approximation to $\frac{ d^j \log \hat Z(y) }{ d y^j } $ for $j= 1,\dots,k$.  We can write
\begin{align*}
   \frac{ d^k \log \hat Z(y) }{ d y^k }  &=  \frac{ d^k }{ d y^k }  \sum_{i=1}^N \log (1- r_i y) = - (k-1)! \sum_{i=1}^N \frac{r_i^k}{(1-r_i y)^k } 
   \\&= - (k-1)! \sum_{i=1}^N r_i^k \sum_{s=0}^{\infty}  \binom{k-1+s } {k-1  } (r_iy)^s \, .
\end{align*}
Now setting 
\[T_t^{(k)} =  -(k-1)!  \sum_{i=1}^N r_i^k \sum_{s=0}^{t}  \binom{k-1+s } {k-1  } r_i^s  \]
we have 
\begin{align*}
    \left|T_t^{(k)} -  \frac{ d^k \log \hat Z(y) }{ d y^k }  \right| &\le  (k-1)! N \sum_{s = t+1}^{\infty} (s+k)^k \eta^s = O (n \eta ^t) \,,
\end{align*}
and so for $t = \Omega_{\Delta, \delta}(\log(1/\epsilon) + k^2)$, the truncation error  can be made at most $\epsilon  n$ as desired.   Moreover, using the algorithm of Patel--Regts~\cite{patel2017deterministic}, $T_t^{(k)}$ can be computed in time  $n e^{O_{\Delta,\delta}(t)} = O_{\Delta, \delta, k}(n(1/\epsilon)^{O_{\Delta, \delta}(1)})$ for this choice of $t$. 

The FPTAS for the mean and variance follow from the additive approximations and \cref{lem:var-bound}. \qedhere 
\end{proof}

\bibliographystyle{amsplain0.bst}
\bibliography{main.bib}

\end{document}